\documentclass[
]{ceurart}

\usepackage{listings}
\usepackage{xcolor}
\usepackage{stmaryrd}
\usepackage{rotating}
\usepackage{tabularx}
\usepackage{booktabs}
\usepackage{amsthm}
\newtheorem{theorem}{Theorem}

\newtheorem{proposition}{Proposition}

\newtheorem{definition}{Definition}

\theoremstyle{definition}
\newtheorem{example}{Example}
\theoremstyle{remark}

\newtheorem*{note*}{Note}
\newtheorem{remark}{Remark}
\newtheorem*{remark*}{Remark}
\theoremstyle{claimstyle}

\newtheorem*{claim*}{Claim}

\newcommand{\shade}[1]{\colorbox{gray!25}{#1}}
\newcommand{\mshade}[1]{\colorbox{gray!25}{$#1$}}
\newcommand{\eval}[1]{$\llbracket$#1$\rrbracket$}
\newcommand{\meval}[1]{\llbracket#1\rrbracket}

\begin{document}

\copyrightyear{2022}
\copyrightclause{Copyright for this paper by its authors.
  Use permitted under Creative Commons License Attribution 4.0
  International (CC BY 4.0).}

\conference{17th Pragmatics of SAT international workshop,
  July 19, 2026, Lisbon, Portugal}

\title{Power Term Polynomial Algebra for Boolean Logic}


\author[1,2]{Emanuele Sansone}[%
orcid=0000-0002-6143-1619,
email=emanuele.sansone@kuleuven.be,
url=https://emsansone.github.io,
]
\cormark[1]
\address[1]{CSAIL, MIT, USA}
\address[2]{ESAT, KU Leuven, Belgium}

\author[1]{Armando Solar-Lezama}[%
orcid=0000-0001-7604-8252,
email=asolar@csail.mit.edu,
url=https://people.csail.mit.edu/asolar/,
]

\cortext[1]{Corresponding author.}

\begin{abstract}
  We introduce \emph{power term polynomial algebra}, a representation language for Boolean formulae designed to bridge conjunctive normal form (CNF) and algebraic normal form (ANF). The language is motivated by the \emph{tiling mismatch} between these representations: direct CNF$\leftrightarrow$ANF conversion may cause exponential blowup unless formulas are decomposed into smaller fragments using auxiliary variables and side constraints. In contrast, our framework addresses this mismatch within the representation itself, compactly encoding structured families of monomials while representing CNF clauses directly, thereby avoiding auxiliary variables and constraints at the abstraction level. We formalize the language through power terms and power term polynomials, define their semantics, and show that they admit algebraic operations corresponding to Boolean polynomial addition and multiplication. We prove several key properties of the language: disjunctive clauses admit compact canonical representations; power terms support local shortening and expansion rewrite rules; and products of atomic terms can be systematically rewritten within the language. Together, these results yield a symbolic calculus that enables direct manipulation of formulae without expanding them into ordinary ANF. The resulting framework provides a new intermediate representation and rewriting calculus that bridges clause-based and algebraic reasoning and suggests new directions for structure-aware CNF$\leftrightarrow$ANF conversion and hybrid reasoning methods.
\end{abstract}


\maketitle

\section{Introduction}
\label{sec:intro}
Representations based on CNF and ANF lie at the core of modern approaches to reasoning about Boolean formulae. CNF representations underpin the success of modern SAT solvers and have been widely adopted for automated reasoning across a broad range of applications. In contrast, ANF representations play a central role in algebraic proof systems and have long been studied as a natural framework for analysing algebraic structure and complexity~\cite{buss2021proof,biere2021handbook}.

Although CNF and ANF describe the same underlying Boolean functions, they organize information in fundamentally different ways and therefore support different styles of reasoning. This difference becomes particularly relevant when converting between the two representations. While such conversions are always possible in principle, they can lead to significant growth in representation size because structure that is compact in one form may expand in the other~\cite{horacek2020conver}. As a result, existing approaches often require decomposing formulae into smaller fragments before translation. In this paper, we refer to this phenomenon as the \emph{tiling mismatch}: the structure of the original representation does not align with the granularity naturally supported by the target representation.

The goal of this work is to introduce a representation language that helps bridge this mismatch. We propose \emph{power term polynomial algebra}, an abstraction language that compactly represents structured families of monomials while still encoding CNF clauses directly. The key idea is to group related monomials into symbolic objects called \emph{power terms}, which can then be combined into \emph{power term polynomials}. This yields a representation that sits between clause-based and polynomial-based descriptions while supporting algebraic-style manipulation.

We first review the Boolean representations underlying SAT and the standard conversions between them, using this discussion to motivate the notion of tiling mismatch. We then introduce power terms and power term polynomials, define their semantics, and show that they admit natural composition rules based on addition and multiplication. Finally, we present rewriting principles that allow expressions to be manipulated directly in the abstraction language without expanding them into ordinary polynomial form.

\section{Preliminaries}
\label{sec:preliminaries}
We introduce the main representations underlying proof systems for the SAT problem, based on resolution and algebraic methods~\cite{buss2021proof}. We then present conversions between these representations and finally introduce the \textit{tiling mismatch} between representations.

In the following, we consider a finite set of variables $\mathcal{D}:=\{x_1,\ldots,x_n\}$, each defined over the set of Boolean values $\mathbb{B}:=\{0,1\}$, where $0$ and $1$ conventionally denote \texttt{false} and \texttt{true}, respectively. It is also useful to introduce the Boolean field $\mathbb{F}_2:=GF(2)$, also known as the Galois field with two elements, whose operations include addition, defined as integer addition modulo 2, and integer multiplication. Finally, we use the compact notation $[n]$ to denote the index set $\{1,\ldots,n\}$.

\noindent\textbf{Boolean functions}. The first representation is based on multivariate functions mapping $\mathbb{B}^n$ to $\mathbb{B}$. The behaviour of a Boolean function $f$ can be specified using a \textit{truth table}, which assigns a Boolean value to each possible combination of variable assignments. A truth table can be represented canonically by listing the assignments for which $f$ evaluates to $1$. Equivalently, this corresponds to defining the set of satisfying assignments $\mathcal{A}(f):=\{a\in\mathbb{B}^n \mid f(a)=1\}$.

\noindent\textbf{Propositional logic formulae}. The second representation introduces logical connectives, including negation $\neg$, conjunction $\land$, and disjunction $\lor$, to construct expressions composed of variables from the domain $\mathcal{D}$. In the context of SAT, formulae are typically expressed in CNF, i.e., as a conjunction of a finite set of clauses $\phi := \bigwedge_{i=1}^{m} C_i$. Each clause is a disjunction of literals $C_i := \bigvee_{j=1}^{n_i} L_{i,j}$, where literals are either variables or their negations (referred to as positive and negative literals, respectively).

Semantically, each logical connective is associated with a predefined truth table. The meaning of clauses and formulae can therefore be defined recursively from variable assignments and the corresponding truth tables of the connectives. Analogously to Boolean functions, the set of satisfying assignments of $\phi$ is the set of variable assignments for which the formula evaluates to \texttt{true}, i.e. $\mathcal{A}(\phi):=\{a\in\mathbb{B}^n \mid \phi(a)=1\}$. Every CNF formula $\phi$ corresponds to a unique Boolean function $f$, defined by the relation $\mathcal{A}(\phi) = \mathcal{A}(f)$. The mapping from CNF formulae to Boolean functions is therefore surjective, since the same Boolean function can be represented by many different CNF formulae.

\noindent\textbf{Boolean polynomials}. The third representation uses polynomials over the field $\mathbb{F}_2$, together with the idempotent relation $x_i^2 = x_i$ for every variable.\footnote{Mathematically written as $\mathbb{F}_2[x_1,\ldots,x_n]/\{x_i^2 = x_i \mid i \in [n]\}$.} This relation ensures that variables appearing in monomials are square-free.

More concretely, a polynomial is a linear combination of monomials
$p:=\sum_{S\subseteq[n]}\alpha_S m_S$, where each monomial is the product of a set of variables $m_S:=\prod_{i\in S}x_i$ (with the convention that $m_\emptyset:=1$) and $\alpha_S \in \mathbb{B}$ are coefficients associated with each monomial. Since every polynomial contains a unique set of monomials and monomials represent a basis for the polynomial space, the relation between the semantic behaviour of Boolean polynomials and Boolean functions is bijective. Consequently, Boolean polynomials provide a canonical representation of propositional logic formulae and are therefore often referred to as their ANF. Another subtle but important result is that the coefficients of the ANF representation are uniquely determined by the truth table of the corresponding Boolean function. In particular, the mapping between the truth table of a Boolean function and the coefficients of its ANF polynomial is given by the M\"obius transform~\cite{rota1964foundations,barbier2019mobius}.

\begin{example}\label{representations}
Let $\mathcal{D}=\{x_1,x_2\}$ and consider the Boolean function $f:\mathbb{B}^2\rightarrow\mathbb{B}$ with satisfying assignments $\mathcal{A}(f)=\{(1,0),(0,1)\}$. Equivalently, $f$ can be represented by the truth table
\[
\begin{array}{c c | c}
x_1 & x_2 & f(x_1,x_2) \\
\hline
0 & 0 & 0 \\
0 & 1 & 1 \\
1 & 0 & 1 \\
1 & 1 & 0
\end{array}
\]
An equivalent CNF formula is $\phi=(x_1 \lor x_2)\land(\neg x_1 \lor \neg x_2)$, for which $\mathcal{A}(\phi)=\mathcal{A}(f)$. The corresponding ANF is $p=x_1+x_2$,
with $\alpha_{\{1\}}=\alpha_{\{2\}}=1$ and $\alpha_{\emptyset}=\alpha_{\{1,2\}}=0$.
These coefficients can be obtained from the truth table values through the M\"obius transform, namely
\[
\begin{pmatrix}
\alpha_{\emptyset}\\
\alpha_{\{1\}}\\
\alpha_{\{2\}}\\
\alpha_{\{1,2\}}
\end{pmatrix}
=
\begin{pmatrix}
M_1 & 0\\
M_1 & M_1
\end{pmatrix}
\begin{pmatrix}
f(0,0)\\
f(1,0)\\
f(0,1)\\
f(1,1)
\end{pmatrix}
\quad (\text{over } \mathbb{F}_2), \qquad 
M_1:=
\begin{pmatrix}
1&0\\
1&1
\end{pmatrix}.
\]
\end{example}

\subsection{Conversions and the Tiling Mismatch}\label{sec:conversion}
\noindent\textbf{Boolean function$\leftarrow$CNF, ANF.} The truth table of a CNF or ANF formula can be obtained by evaluating the formula on all assignments. In the case of ANF, the truth table can also be recovered by applying the inverse M\"obius transform to the polynomial coefficients.

\noindent\textbf{Boolean function$\rightarrow$CNF, ANF.} A Boolean function can be converted to CNF via the product of maxterms. Specifically, the rows of the truth table for which $f$ evaluates to $0$ are selected, and each row is translated into a clause (maxterm) where a variable appears unnegated if its value in the row is $0$, and negated if its value is $1$. The conjunction of all such clauses yields a CNF formula equivalent to $f$. This construction is also illustrated in \autoref{representations}. The conversion to ANF is obtained by expressing $f$ as a Boolean polynomial whose coefficients are computed from the truth table of $f$ via the M\"obius transform. Notably, if a Boolean function is constant, its ANF has a particularly simple form: if $f$ is $0$, then all coefficients are $0$, while if $f$ is $1$, only the coefficient $\alpha_{\emptyset}$ is equal to $1$. In propositional terms, these two cases correspond to an unsatisfiable formula and a tautology, respectively.

\noindent\textbf{CNF$\rightarrow$ANF.} Two standard strategies for this conversion are described in~\cite{horacek2020conver}. The first strategy~\cite{hsiang1985refutational} performs a \emph{direct} conversion by translating each clause into its ANF representation and then combining the polynomials from all clauses through multiplication, thus producing a final factorized polynomial. To compute $\mathrm{ANF}(C_i)$, it is convenient to first consider the negation $\mathrm{ANF}(\neg C_i)$. By De Morgan's law, the negated clause $\neg C_i=\bigwedge_{j=1}^{n_i}\neg L_{i,j}$ can be expressed as the product of the ANF representations of the negated literals in the clause. For a positive literal $L_{i,j}=x$, the negation $\neg L_{i,j}$ corresponds to the polynomial $1+x$, while for a negative literal $L_{i,j}=\neg x$, the negation simplifies to $x$. Expanding the resulting product and adding the constant term $1$ yields $\mathrm{ANF}(C_i)$. However, the number of monomials in the resulting ANF grows exponentially with the number of positive literals.

The second strategy, proposed in~\cite{alekhnovich2002space}, avoids this exponential blowup by exploiting the fact that clauses containing only negative literals translate into very small ANF polynomials, whereas clauses with many positive literals produce polynomials with exponentially many monomials. The strategy therefore introduces an encoding over an expanded domain of variables in which positive literals are replaced by auxiliary \emph{twin variables}. These variables are constrained so that $x_i' = 1 + x_i$, allowing positive literals to be rewritten as negative ones. The direct conversion can then be applied to the transformed formula, whose clauses contain only negative literals (or significantly fewer positive ones), yielding a factorized polynomial with small factors. In the worst case, this approach can double both the number of variables and the number of resulting polynomials. A practical illustration of these two strategies is provided in \autoref{conversions}. We refer the interested reader to~\cite{horacek2020conver} for more advanced conversions building on top of these two strategies.

\begin{example}\label{conversions}
    Consider the formula $\phi=(x_1\lor x_2)\land(\neg x_1\lor\neg x_2\lor \neg x_4)\land(\neg x_1\lor\neg x_3\lor x_4)\land(x_1\lor x_2 \lor x_3)$.
    The first strategy, i.e. \emph{direct} conversion, performs the following sequence of steps (fully shown only for the first three clauses):
    \begin{align}
        (x_1\lor x_2)
        &\rightarrow \neg(\neg x_1\land \neg x_2)
        \rightarrow 1+((1+x_1)(1+x_2))
        \rightarrow x_1+x_2+x_1x_2, \nonumber \\
        (\neg x_1\lor\neg x_2\lor \neg x_4)
        &\rightarrow \neg(x_1\land x_2 \land x_4)\rightarrow 1+(x_1x_2x_4)
        \rightarrow 1+x_1x_2x_4. \nonumber\\
        (\neg x_1\lor \neg x_3\lor x_4)
        &\rightarrow \neg(x_1\land x_3\land \neg x_4)
        \rightarrow 1+(x_1x_3(1+x_4))
        \rightarrow 1+x_1 x_3+x_1 x_3x_4, \nonumber\\
        (x_1\lor x_2\lor x_3)
        &\rightarrow \ldots
        \rightarrow \ldots
        \rightarrow x_1+x_2+x_3+x_1x_2+x_1x_3+x_2x_3+x_1x_2x_3, \nonumber
        \end{align}
    thus resulting in the factorized Boolean polynomial $p=(x_1+x_2+x_1x_2)(1+x_1x_2x_4)(1+x_1 x_3+x_1 x_3x_4)(x_1+x_2+x_3+x_1x_2+x_1x_3+x_2x_3+x_1x_2x_3)$. Observe that clauses containing only negative literals translate into very small polynomials, whereas clauses with many positive literals expand into polynomials with many monomials. This difference motivates the second strategy, which replaces positive literals with twin variables so that most clauses become negative.
    
    The second approach introduces the set of twin variables $\{x_i':=1+x_i\mid \forall i\in[4]\}$ which enable to equivalently rewrite $\phi$ using a new formula $\phi'=(\neg x_1'\lor \neg x_2')\land(\neg x_1\lor\neg x_2\lor \neg x_4)\land(\neg x_1\lor \neg x_3\lor\neg x_4')\land(\neg x_1'\lor \neg x_2'\lor\neg x_3')$ with no positive literals. Now, by leveraging the \emph{direct} conversion on $\phi'$, one obtains the polynomial $p'=(1+x_1'x_2')(1+x_1x_2x_4)(1+x_1x_3x_4')(1+x_1'x_2'x_3')$. Importantly, additional polynomials are included in $p'$ to enforce the constraints on the twin variables, yielding the final polynomial $p''= p'(x_1+x_1')(x_2+x_2')(x_3+x_3')(x_4+x_4')$. While $p''$ has fewer monomials in each factor compared to $p$, the number of variables and factors increases due to the introduction of the twin variables. Notably, by substituting the twin variables according to the constraints $x_i' = 1 + x_i$ in $p''$ (or directly in $p'$), one can verify that $p'',p',p$ all simplify to the same Boolean polynomial.
\end{example}

\noindent\textbf{CNF$\leftarrow$ANF.} Two standard strategies for this conversion are described in~\cite{horacek2020conver}. The first strategy is \emph{naive} in the sense that it converts the polynomial into an explicit Boolean function before deriving the CNF. Concretely, it performs the transformation ANF$\rightarrow$Boolean, followed by the Boolean$\rightarrow$CNF conversion. Clearly, this approach does not scale well with the number of variables.

The second strategy addresses the scalability issue by decomposing the problem into smaller, more tractable subproblems, to which the naive strategy can still be applied. To achieve this decomposition, the domain of variables is expanded in two stages. In the first stage, auxiliary variables are introduced to replace all monomials of degree greater than one, thereby linearizing the original polynomial.\footnote{Linearization enables the use of Gaussian elimination to further simplify the polynomials.} In the second stage, additional auxiliary variables are introduced to split the polynomial into smaller chunks / tiles, reducing the number of monomials and enabling the naive conversion. The definitions of the auxiliary variables are encoded through additional clauses, which contribute to the total number of clauses in the resulting CNF. In the worst case, this number grows exponentially with the size of the chunks~\cite{horacek2020conver}. A practical illustration of the domain-expansion strategy is shown in \autoref{inverse_conversions}.

\begin{example}\label{inverse_conversions}
    Given a polynomial $p=1+x_1+x_2+x_3+x_1 x_2$, linearization produces the linear polynomial $p'=1+x_1+x_2+x_3+x_1'$ by introducing the auxiliary variable $x_1':=x_1 x_2$. The constraint on $x_1'$ is converted to the CNF $(x_1'\lor \neg x_1\lor\neg x_2)\land(\neg x_1'\lor x_1)\land(\neg x_1'\lor x_2)$. Subsequently, polynomial $p'$ can be split into chunks / tiles of arbitrary size $k$ ($k=3$ in this example) through the additional variable $x_1'':=x_3+x_1'$. This yields two polynomials $p''=1+x_1+x_2+x_1''$ and $p'''=1+x_1''+x_3+x_1'$, where the equation $p'''=1$ encodes the definition of $x_1''$. After splitting, the polynomials $p'',p'''$ are converted into CNF using the naive strategy. All three CNFs, i.e. the one constraining $x_1'$ and the ones from $p'',p'''$, are finally conjoined together to produce the final representation.
\end{example}

\begin{remark}
 Both ANF$\rightarrow$CNF and ANF$\leftarrow$CNF conversions
 may incur an exponential blowup. In order to avoid this, 
 the problem needs to be tiled into
 smaller fragments that can be converted
 independently. This is done through 
 the introduction of auxiliary variables and corresponding constraints. Before the problem has been decomposed in this way, we say there is a \emph{tiling mismatch} between the current form 
 of the problem and tile size that can be supported by the conversion. Importantly, the introduction of variables and constraints needed to handle the tiling mismatch introduces a computational overhead, both in terms of memory and time. This is critical for SAT solvers or algebraic proof systems~\cite{biere2021handbook,bright2022sat}.  
\end{remark}

\section{Power Term Polynomial Algebra}
\label{sec:power}
We introduce an abstraction language to address the \emph{tiling mismatch}. This language (i) compactly represents Boolean polynomials through monomial grouping and (ii) encodes CNF formulae without introducing additional variables or constraints. It also provides composition rules based on addition and multiplication, enabling direct manipulation of expressions within the abstract representation.
We begin by introducing the fundamental building blocks of the language, \textit{power terms}. We then show how these can be composed to form \textit{power term polynomials}. Finally, we establish that power term polynomials form an algebra and present rules for manipulating them directly within this representation space.

\subsection{Power Terms}
\begin{figure}[t]
\centering
\begin{tabular}{lllllll}
\textbf{Boolean Poly.} & & \textbf{Set of Sets} & & \textbf{Power Set} & & \textbf{Power Term} \\
$x_1$ &$\rightarrow$ & $\{\{x_1\}\}$ & $\rightarrow$ & $\mathbf{P}_1$ & $\rightarrow$ & $S_1.\mathbf{P}_\emptyset$\\
$x_1x_2$ &$\rightarrow$ & $\{\{x_1,x_2\}\}$ & $\rightarrow$ &  $\times$ & $\rightarrow$ & $S_{1,2}.\mathbf{P}_\emptyset$\\
$x_1{+}x_2{+}x_1x_2$ & $\rightarrow$ & $\{\{x_1\},\{x_2\},\{x_1,x_2\}\}$  & $\rightarrow$ & $\mathbf{P}_{1,2}$ & $\rightarrow$ & $S_\emptyset.\mathbf{P}_{1,2}$\\
$x_3(x_1{+}x_2{+}x_1x_2)$ & $\rightarrow$ & $\{\{x_1{,}x_3\}{,}\{x_2{,}x_3\}{,}\{x_1{,}x_2{,}x_3\}\}$  & $\rightarrow$ & $\times$ & $\rightarrow$ & $S_3.\mathbf{P}_{1,2}$
\end{tabular}
\caption{Abstraction steps illustrating the transformation of Boolean polynomials into power terms. The symbol $\times$ indicates that the corresponding set of sets representation cannot be expressed as a single power set.}
\label{ex:terms}
\end{figure}
Our goal is to represent Boolean polynomials compactly as a single structured term. To this end, we interpret each monomial as the set of variables appearing in it. For example, the monomial $x_1x_2$ is represented by the set $\{x_1,x_2\}$. Under this interpretation, a Boolean polynomial becomes a set of sets, where each inner set corresponds to a monomial.

This viewpoint reveals a combinatorial structure closely related to power sets, which we exploit to obtain a more compact abstraction. The correspondence is illustrated in \autoref{ex:terms}.
\begin{example}
Consider the Boolean polynomials and their corresponding set of sets representations shown in the two leftmost columns of \autoref{ex:terms}. The first two rows correspond to single monomials. The last two rows illustrate families of monomials corresponding to subsets of the variable set $\{x_1,x_2\}$. For instance, the polynomial $x_1 + x_2 + x_1x_2$ corresponds to the set $\{\{x_1\},\{x_2\},\{x_1,x_2\}\}$ which is precisely the power set of $\{x_1,x_2\}$ without the empty set.
\end{example}

Since the standard power set includes the empty set, we use a modified variant that excludes it. This matches the structure of our set of sets representation, where the empty monomial does not appear. Bold notation is used to distinguish power sets from ordinary sets.

\begin{definition}[Power Set]\label{power_set} For a set $U\subseteq\mathcal{D}$, the (modified) power set is defined as
\[
\mathbf{P}_U :=\begin{cases}
\{\emptyset\} & \text{if } U = \emptyset,\\
\{T \subseteq U \mid T \neq \emptyset\} & \text{otherwise}.
\end{cases}
\]
\end{definition}
\begin{example}
\autoref{power_set} allows several of the set of sets representations in \autoref{ex:terms} to be expressed more compactly. In particular, the expressions $\{\{x_1\}\}$, $\{\{x_1\},\{x_2\},\{x_1,x_2\}\}$ can be written as $\mathbf{P}_1$, and $\mathbf{P}_{1,2}$, respectively. For readability, we write $\mathbf{P}_{i_1,\ldots,i_k}$ instead of $\mathbf{P}_{\{x_{i_1},\ldots,x_{i_k}\}}$ when the underlying variables are clear. Importantly, expressions such as $\{\{x_1,x_2\}\}$ and $\{\{x_1,x_3\}{,}\{x_2,x_3\}{,}\{x_1,x_2,x_3\}\}$ cannot be represented as a single power set.
\end{example}

Power sets alone are not sufficient, since we must still distinguish single monomials from structured families of monomials. This motivates the notion of a power term.

\begin{definition}[Power Term]\label{power_term}
Let $S,U\subseteq\mathcal{D}$ such that
\begin{enumerate}
    \item $S\cap U=\emptyset$,
    \item $|U|\neq 1$,
    \item $(S,U)\neq(\emptyset,\emptyset)$.
\end{enumerate}
The power term generated by $S$ and $\mathbf{P}_U$ is the set: 
\[
S.\mathbf{P}_U:=\{S\cup T\mid T\in\mathbf{P}_U\}
\]
The degree of the power term $S.\mathbf{P}_U$ is defined as $|U|$. 
\end{definition}
Condition~1 ensures that the base set $S$ does not overlap with the variables appearing in the combinations generated by $\mathbf{P}_U$.
Condition~2 prevents the case $|U|=1$. Indeed, if $|U|=1$ the generated family contains exactly one monomial, which could equivalently be represented as a degree-0 power term by moving the element of $U$ into $S$. Intuitively, variables in $S$ appear in every generated monomial, while variables in $U$ may appear in any non-empty subset.
Condition~3 excludes the trivial case where both $S$ and $U$ are empty.
\begin{example}
The first two monomials, or equivalently set of sets representations, in \autoref{ex:terms} correspond to single monomials and can therefore be expressed using power terms of degree $0$. 
In contrast, the last two compound expressions correspond to families of monomials and are represented using power terms of degree greater than $1$, in accordance with Condition~2 in \autoref{power_term}.
\end{example}
\begin{figure}
    \centering
    \begin{tabular}{rcl|rll}
     \multicolumn{3}{c}{\textbf{Syntax}} & \multicolumn{3}{c}{\textbf{Evaluation}} \\
     \\
     \shade{\texttt{expr}}   & $::=$ & & \shade{\eval{\texttt{expr}}} & = & \\
     & & \mshade{\texttt{expr}\;\odot\;\texttt{expr}} & & & \shade{\eval{\texttt{expr}}$\cdot$\eval{\texttt{expr}}}\\
     & & \shade{\texttt{ptpoly}} & & & \shade{\eval{\texttt{ptpoly}}}\\
     \shade{\texttt{ptpoly}}   & $::=$ & & \shade{\eval{\texttt{ptpoly}}} & = & \\
     & & \mshade{\texttt{ptpoly}\;\uplus\;\texttt{ptpoly}} & & & \shade{\eval{\texttt{ptpoly}}$+$\eval{\texttt{ptpoly}}}\\
    & & \mshade{S_I\;.\;\mathbf{P}_\emptyset} & & & \mshade{1}\\     
    & & \mshade{S_\emptyset\;.\;\mathbf{P}_\emptyset} & & & \mshade{0}\\
     & & \shade{\texttt{pterm}} & & & \shade{\eval{\texttt{pterm}}}\\
     \shade{\texttt{pterm}}     & $::=$ & & \shade{\eval{\texttt{pterm}}} & = & \\
     & & \mshade{S\;.\;\mathbf{P}_U\;\text{ (with }S,U\subseteq\mathcal{D})} & & & \shade{$\sum_{T\in S.\mathbf{P}_U}\prod_{x\in T}x$}
    \end{tabular}
    \caption{Syntax and (denotational) semantics of the abstraction language based on power term polynomials. A power term polynomial generated by this grammar is well formed if every power term $S.\mathbf{P}_U$ satisfies the conditions of Definition~\ref{power_term}. The operators $\uplus$ and $\odot$ satisfy the algebraic laws given in their respective definitions and are interpreted as addition and multiplication over $\mathbb{F}_2$, respectively.}
    \label{fig:language}
\end{figure}

Having introduced power terms, we now describe how they are represented within the abstraction language. Figure~\ref{fig:language} presents the corresponding grammar. A \texttt{pterm} has the form $S.\mathbf{P}_U$, where $S,U\subseteq\mathcal{D}$ satisfy the conditions of \autoref{power_term}. 
Its semantics is the Boolean polynomial obtained by enumerating all monomials generated by the family $S.\mathbf{P}_U$, i.e., each monomial contains all variables in $S$ together with a non-empty subset of variables from $U$.

\begin{example}
The power term $S_\emptyset.\mathbf{P}_{1,2}$ represents $x_1 + x_2 + x_1x_2$, while 
$S_3.\mathbf{P}_{1,2}$ represents $x_1x_3 + x_2x_3 + x_1x_2x_3$. 
Single monomials correspond to degree-$0$ power terms, e.g., $S_{1,2}.\mathbf{P}_\emptyset$ represents $x_1x_2$.
\end{example}

\subsection{Power Term Polynomials}
Power terms provide a compact representation of structured families of monomials. However, a single power term cannot in general represent an arbitrary Boolean polynomial. To obtain a complete representation language, we therefore allow power terms to be combined using \emph{exclusive disjunction}. The resulting expressions are called \emph{power term polynomials}. Intuitively, a power term polynomial represents a Boolean polynomial as the addition over $\mathbb{F}_2$ of several structured families of monomials.

\begin{definition}[Power Term Polynomial]\label{power_term_polynomial}
Let $\mathcal{T}$ denote the set of all well-formed power terms. In addition, we introduce the constants $S_\emptyset.\mathbf{P}_\emptyset$ and
$S_I.\mathbf{P}_\emptyset$, written in the same syntactic form as power terms but treated as constants of the language. These constants denote the Boolean constants $0$ and $1$, respectively, and are part of the language but are not
required to satisfy the well-formedness conditions of
Definition~\ref{power_term}.

A power term polynomial corresponds to the syntactic category
\texttt{ptpoly} of the grammar shown in \autoref{fig:language}. Concretely, a power term polynomial is an expression generated by the grammar
\[
\pi ::= S_\emptyset.\mathbf{P}_\emptyset \mid S_I.\mathbf{P}_\emptyset \mid t \mid \pi \uplus \pi
\]
where $t \in \mathcal{T}$. The binary operator $\uplus$ is a syntactic constructor satisfying the following equational axioms:
\begin{align}
    \pi \uplus \rho &= \rho \uplus \pi && \text{(Commutativity)} \nonumber\\
    (\pi \uplus \rho) \uplus \sigma &= \pi \uplus (\rho \uplus \sigma) && \text{(Associativity)} \nonumber\\
    \pi \uplus S_\emptyset.\mathbf{P}_\emptyset &= \pi && \text{(Additive Identity)} \nonumber\\
    \pi \uplus \pi &= S_\emptyset.\mathbf{P}_\emptyset && \text{(Additive Inverse)} \nonumber
\end{align}
for all power term polynomials $\pi,\rho,\sigma$.
\end{definition}

Thus power terms serve as the atomic generators of the language, while more complex expressions are obtained by combining them through exclusive disjunction. It follows immediately that power term polynomials are closed under exclusive disjunction. The semantics of a power term polynomial is obtained by evaluating each power term and combining the resulting Boolean polynomials using addition over $\mathbb{F}_2$:
\[
\meval{\pi \uplus \rho} = \meval{\pi} + \meval{\rho}, \qquad\meval{S_\emptyset.\mathbf{P}_\emptyset}=0, \qquad \meval{S_I.\mathbf{P}_\emptyset}=1.
\]
This is also summarized in \autoref{fig:language}.
Since multiple power term polynomials may represent the same Boolean polynomial, we measure the complexity of a representation by counting its power terms.

\begin{definition}[Size, Minimality]
Let $\pi$ be a power term polynomial. The \emph{size} of $\pi$, denoted as $\mathrm{Size}(\pi)$, is the number of atomic terms (power terms or constants) occurring in $\pi$. A power term polynomial $\pi$ representing a Boolean polynomial $p$ is minimal if for every power term polynomial $\rho$ representing $p$, $\mathrm{Size}(\pi)\leq\mathrm{Size}(\rho)$.
\end{definition}

\begin{example}
Consider the Boolean polynomial $p = 1 + x_1 + x_2 + x_1x_2$. This can be encoded using different power term polynomials. For instance, $\pi:=S_I.\mathbf{P}_\emptyset\uplus S_1.\mathbf{P}_\emptyset\uplus S_2.\mathbf{P}_\emptyset\uplus S_{1,2}.\mathbf{P}_\emptyset$ and $\rho:=S_I.\mathbf{P}_\emptyset\uplus S_\emptyset.\mathbf{P}_{1,2}$. Both $\pi,\rho$ are semantically equivalent, with $\mathrm{Size}(\pi)=4$ and $\mathrm{Size}(\rho)=2$. Notably, $\rho$ is minimal.
\end{example}

\paragraph{Properties of Power Term Polynomials}
We now describe three structural properties of power term polynomials that characterize their expressive power and limitations.

\begin{proposition}[Property 1: Lack of General Canonicity]\label{prop:noncanonical}
Minimal power term polynomials are not canonical with respect to Boolean polynomials.
\end{proposition}

\begin{proof}
We prove this by counterexample. Consider the Boolean polynomial $p := x_1 + x_2$. Two minimal representations are $\pi:=S_1.\mathbf{P}_\emptyset \uplus S_2.\mathbf{P}_\emptyset$ and $\rho:=S_\emptyset.\mathbf{P}_{1,2} \uplus S_{1,2}.\mathbf{P}_\emptyset$. Both contain two power terms and are therefore minimal representations.

As a second example, consider $p' := x_1 + x_2 + x_3 + x_1x_2 + x_1x_3$. Two minimal representations are $\pi':=S_\emptyset.\mathbf{P}_{1,2,3} \uplus S_{1,2,3}.\mathbf{P}_\emptyset \uplus S_{2,3}.\mathbf{P}_\emptyset$ and  $\rho':=S_\emptyset.\mathbf{P}_{1,2} \uplus S_{3}.\mathbf{P}_\emptyset \uplus S_{1,3}.\mathbf{P}_\emptyset$.
Both contain three power terms and are therefore minimal.
Hence minimal power term polynomials are not unique.
\end{proof}

\begin{proposition}[Property 2: Canonicity of Disjunctive Formulae]\label{prop:canonical-disjunction}
For any disjunctive formula $\phi$, the corresponding minimal power term polynomial $\pi$ provides a canonical construction for $\phi$. Moreover, $\mathrm{Size}(\pi)\leq 3$, independently of the number of positive literals (conversion rules from a disjunctive formula to a power term polynomial are given in the proof).
\end{proposition}

\begin{proof}
Consider a general disjunctive formula, with positive and negative literals grouped on the left and right sides of the formula, respectively.
\[
\phi:=\bigvee_{i=1}^k x_i \lor\bigvee_{j=k+1}^{k+\ell} \neg x_j
\]
We first convert $\phi$ into a Boolean polynomial $p$ using the direct conversion mentioned in \autoref{sec:conversion}.
\[
p=1+\bigg(\prod_{i=1}^k(1+x_i)\bigg)\cdot\bigg(\prod_{j=k+1}^{k+\ell} x_j\bigg)
\]
We distinguish three cases:
\smallskip
\begin{itemize}
    \item \textbf{Case 1} ($\ell=0$): $p=1+\prod_{i=1}^{k} (1+x_i)$. By expanding the polynomial, we are left with the sum of all possible combinations of variables in $\{x_{1},\ldots,x_{k}\}$. This can be captured using power set notation. In other words,  $p=\sum_{U\in\mathbf{P}_{1,\ldots,k}}\prod_{i\in U}x_i$. The Boolean polynomial $p$ can be minimally encoded using $\pi=S_\emptyset.\mathbf{P}_{1,\ldots,k}$ for $k>1$ and $\pi=S_1.\mathbf{P}_\emptyset$ for $k=1$. Therefore, $\mathrm{Size}(\pi)=1$.
    \smallskip
    \item \textbf{Case 2} ($k=0$): $p=1+\prod_{j=1}^\ell x_j$. Since the polynomial contains both the constant term and the monomial $\prod_{j=1}^\ell x_j$, any representation requires at least two atomic terms. This admits the minimal representation $\pi=S_I.\mathbf{P}_\emptyset\uplus S_{1,\ldots,\ell}.\mathbf{P}_\emptyset$. Hence, $\mathrm{Size}(\pi)=2$. 
    \smallskip
    \item \textbf{Case 3} ($k>0$ and $\ell>0$): this is the general case. $p$ can be rewritten in the following way:
    \begin{align}
    p&=1+\bigg(\prod_{i=1}^k(1+x_i)\bigg)\cdot\bigg(\prod_{j=k+1}^{k+\ell} x_j\bigg)=1+\bigg(1+\sum_{U\in\mathbf{P}_{1,\ldots,k}}\prod_{i\in U}x_i\bigg)\cdot\bigg(\prod_{j=k+1}^{k+\ell} x_j\bigg) \nonumber\\
    &=1+\prod_{j=k+1}^{k+\ell} x_j+\prod_{j=k+1}^{k+\ell} x_j\sum_{U\in\mathbf{P}_{1,\ldots,k}}\prod_{i\in U}x_i \nonumber
    \end{align}
    The last equality can be minimally encoded using the power term polynomial $\pi=S_I.\mathbf{P}_\emptyset\uplus S_{k+1,\ldots,k+\ell}.\mathbf{P}_\emptyset\uplus S_{k+1,\ldots,k+\ell}.\mathbf{P}_{1,\ldots,k}$ for $k>1$ and $\pi=S_I.\mathbf{P}_\emptyset\uplus S_{2,\ldots,\ell+1}.\mathbf{P}_\emptyset\uplus S_{1,\ldots,\ell+1}.\mathbf{P}_\emptyset$ for $k=1$. Therefore, $\mathrm{Size}(\pi)=3$.
\end{itemize}
This concludes the proof.
\end{proof}

\begin{example}
We illustrate \autoref{prop:canonical-disjunction} using the formula
\[
\phi=(x_1\lor x_2)\land(\neg x_1\lor\neg x_2\lor \neg x_4)\land(\neg x_1\lor\neg x_3\lor x_4)\land(x_1\lor x_2 \lor x_3)
\]
introduced in \autoref{conversions}. By \autoref{prop:canonical-disjunction}, each clause admits a canonical power term polynomial representation of size at most $3$. Concretely, the four clauses are represented as
\begin{align}
x_1\lor x_2 \;&\mapsto\; S_\emptyset.\mathbf{P}_{1,2}, \nonumber\\
\neg x_1\lor\neg x_2\lor \neg x_4 \;&\mapsto\; S_I.\mathbf{P}_\emptyset \uplus S_{1,2,4}.\mathbf{P}_\emptyset, \nonumber\\
\neg x_1\lor\neg x_3\lor x_4 \;&\mapsto\; S_I.\mathbf{P}_\emptyset \uplus S_{1,3}.\mathbf{P}_\emptyset \uplus S_{1,3,4}.\mathbf{P}_{\emptyset}, \nonumber\\
x_1\lor x_2 \lor x_3 \;&\mapsto\; S_\emptyset.\mathbf{P}_{1,2,3}. \nonumber
\end{align}
This representation avoids the introduction of auxiliary variables or constraints, while still providing a compact encoding of all clauses and thereby addressing the \emph{tiling mismatch}.
\end{example}

\begin{proposition}[Property 3: Shortening / Expanding Power Term Polynomials]\label{shorten_expand}
Let $S,T,V\subseteq\mathcal{D}$ be pairwise disjoint sets. Then, the following identities hold:
\[
S.\mathbf{P}_{T\cup V}=
\begin{cases}
    (S\cup T).\mathbf{P}_\emptyset\,\uplus\,(S\cup T\cup V).\mathbf{P}_\emptyset\,\uplus\,(S\cup V).\mathbf{P}_\emptyset & \text{if }|T|=1\text{ and }|V|=1\text{ (Case 1)}\\
    S.\mathbf{P}_T\,\uplus\,(S\cup V).\mathbf{P}_T\,\uplus\,(S\cup V).\mathbf{P}_\emptyset & \text{if }|T|>1\text{ and }|V|=1\text{ (Case 2)}\\
   \biguplus_{\widetilde{V}\in\mathbf{P}_V}(S\cup\widetilde{V}).\mathbf{P}_T\,\uplus\, S.\mathbf{P}_T\uplus S.\mathbf{P}_V & \text{if }|T|>1\text{ and }|V|>1\text{ (Case 3)}
\end{cases}
\]
\end{proposition}
The identities in \autoref{shorten_expand} provide local rewriting rules that allow power term polynomials to be shortened or expanded by regrouping families of
monomials, thereby enabling more compact representations, while preserving semantics.
The proof of \autoref{shorten_expand} is given in \autoref{app_shorten}.

\begin{example}\label{ex:shorten-expand}
We illustrate the three rewriting rules of \autoref{shorten_expand} over the same variable domain $\mathcal{D}=\{x_1,x_2,x_3,x_4\}$.
\medskip

\noindent\textbf{Case 1.} Let $S=\emptyset$, $T=\{x_1\}$, and $V=\{x_2\}$. Then, $T\cup V=\{x_1,x_2\}$, and
$S.\mathbf{P}_{T\cup V}=S_\emptyset.\mathbf{P}_{1,2}=
S_1.\mathbf{P}_\emptyset
\uplus
S_{1,2}.\mathbf{P}_\emptyset
\uplus
S_2.\mathbf{P}_\emptyset$.

\medskip
\noindent\textbf{Case 2.} Let $S=\emptyset$, $T=\{x_1,x_2\}$, and $V=\{x_3\}$. Then, $T\cup V=\{x_1,x_2,x_3\}$, and $S.\mathbf{P}_{T\cup V}=S_\emptyset.\mathbf{P}_{1,2,3}=
S_\emptyset.\mathbf{P}_{1,2}
\uplus
S_3.\mathbf{P}_{1,2}
\uplus
S_3.\mathbf{P}_\emptyset$.
Now observe that 
\[
\meval{S_\emptyset.\mathbf{P}_{1,2,3}}=x_1+x_2+x_3+x_1x_2+x_1x_3+x_2x_3+x_1x_2x_3
\]
and
\[
\meval{S_\emptyset.\mathbf{P}_{1,2}
\uplus
S_3.\mathbf{P}_{1,2}
\uplus
S_3.\mathbf{P}_\emptyset}=(x_1+x_2+x_1x_2)
+
(x_1x_3+x_2x_3+x_1x_2x_3)
+
x_3
\]
are semantically equivalent.

\medskip
\noindent\textbf{Case 3.} Let $S=\emptyset$, $T=\{x_1,x_2\}$, and $V=\{x_3,x_4\}$. Then $T\cup V=\{x_1,x_2,x_3,x_4\}$, and
\[
S.\mathbf{P}_{T\cup V}=S_\emptyset.\mathbf{P}_{1,2,3,4}.
\]
Since $\mathbf{P}_V=\{\{x_3\},\{x_4\},\{x_3,x_4\}\}$, \autoref{shorten_expand} gives
$
S_\emptyset.\mathbf{P}_{1,2,3,4}
=
S_3.\mathbf{P}_{1,2}
\uplus
S_4.\mathbf{P}_{1,2}
\uplus
S_{3,4}.\mathbf{P}_{1,2}
\uplus
S_\emptyset.\mathbf{P}_{1,2}
\uplus
S_\emptyset.\mathbf{P}_{3,4}
$.
Thus one larger family of monomials is decomposed into smaller structured components.

These rewriting rules are useful because they allow one to move between coarser and finer decompositions of the same Boolean polynomial. This is essential when searching for compact encodings: depending on the surrounding expression, expanding a power term may expose cancellation opportunities, while shortening several terms into a single one may reduce the overall size of the representation.
\end{example}

\subsection{Power Term Polynomial Algebra}
So far, power term polynomials have been equipped with the operation
$\uplus$, which captures addition over $\mathbb{F}_2$. To obtain a full
algebraic structure, we now introduce a multiplication operator on power
term polynomials. Intuitively, this operator corresponds to the product
of the Boolean polynomials denoted by the operands. As with $\uplus$, we
define multiplication syntactically through equational axioms and
interpret it semantically via evaluation.

\begin{definition}[Multiplication of Expressions]\label{def:pt-multiplication}
Let $\mathcal{E}$ denote the set of expressions generated by the syntactic
category \texttt{expr} in \autoref{fig:language}. Thus every power term
polynomial is an expression, but expressions may additionally contain
products formed with $\odot$.

The binary operator $\odot$ is a syntactic constructor on expressions
satisfying the following equational axioms for all $e,f,g\in\mathcal{E}$:
\begin{align}
    e \odot f &= f \odot e && \text{(Commutativity)} \nonumber\\
    (e \odot f)\odot g &= e \odot (f \odot g) && \text{(Associativity)} \nonumber\\
    e \odot S_I.\mathbf{P}_\emptyset &= e && \text{(Multiplicative Identity)} \nonumber\\
    e \odot S_\emptyset.\mathbf{P}_\emptyset &= S_\emptyset.\mathbf{P}_\emptyset && \text{(Multiplicative Zero)} \nonumber\\
    e \odot (f \uplus g) &= (e \odot f)\uplus(e \odot g) && \text{(Left Distributivity)} \nonumber\\
    (e \uplus f)\odot g &= (e \odot g)\uplus(f \odot g) && \text{(Right Distributivity)} \nonumber
\end{align}
\end{definition}

The semantics of multiplication is defined by
\[
\meval{e \odot f} = \meval{e}\cdot\meval{f},
\]
where the product on the right-hand side is the usual multiplication of Boolean polynomials over $\mathbb{F}_2$. Thus $\odot$ equips the abstraction language with a multiplicative structure that mirrors the product of Boolean polynomials. In particular, if $\pi$ and $\rho$ are power term polynomials, then $\pi\odot\rho$ is a well-formed expression of the syntactic category \texttt{expr}.

\begin{definition}[Boolean Field of the Abstraction Language]
Let $\mathbb{B}_{\mathrm{PT}} :=
\{S_\emptyset.\mathbf{P}_\emptyset,\; S_I.\mathbf{P}_\emptyset\}$ be the \emph{Boolean field of the
abstraction language}. The operations on $\mathbb{B}_{\mathrm{PT}}$ are
the restrictions of $\uplus$ and $\odot$ to these two constants.
\end{definition}

$\mathbb{B}_{\mathrm{PT}}$ is simply the syntactic realization
inside the abstraction language of the two-element field
$\mathbb{F}_2$.

\begin{proposition}
$\big(\mathbb{B}_{\mathrm{PT}},\uplus,\odot\big)$ is a field.
\end{proposition}

\begin{proof}
The set $\mathbb{B}_{\mathrm{PT}}$ contains exactly the two elements
$S_\emptyset.\mathbf{P}_\emptyset$ and $S_I.\mathbf{P}_\emptyset$,
which denote $0$ and $1$, respectively. By Definition~\ref{power_term_polynomial},
the operators $\uplus$ and $\odot$ correspond to addition and
multiplication over $\mathbb{F}_2$. Hence
$\big(\mathbb{B}_{\mathrm{PT}},\uplus,\odot\big)$ is isomorphic to the
two-element field $\mathbb{F}_2$, and therefore it is a field.
\end{proof}

We now lift this structure from constants to arbitrary expressions of the abstraction language.

\begin{proposition}[Ring Structure of the Abstraction Language]
Let $\mathcal{E}$ be the set of all expressions generated by the syntactic category \texttt{expr} in \autoref{fig:language}. The operation $\uplus$, originally defined on power term polynomials, extends naturally to expressions in $\mathcal{E}$ through the grammar. Then, $(\mathcal{E},\uplus,\odot,S_\emptyset.\mathbf{P}_\emptyset,S_I.\mathbf{P}_\emptyset)$ is a commutative ring with unity.
\end{proposition}

\begin{proof}
By Definition~\ref{power_term_polynomial}, the operation $\uplus$ is commutative and associative on power term polynomials, admits $S_\emptyset.\mathbf{P}_\emptyset$ as neutral element, and every element is its own inverse. These laws extend to all expressions in $\mathcal{E}$ through the grammar of \autoref{fig:language}.

By Definition~\ref{def:pt-multiplication}, multiplication $\odot$ is associative and commutative, with identity
$S_I.\mathbf{P}_\emptyset$ and absorbing element $S_\emptyset.\mathbf{P}_\emptyset$. The distributive laws ensure that $\odot$ distributes over $\uplus$ on both sides. Therefore, $\mathcal{E}$ is a commutative ring with unity.
\end{proof}

The preceding proposition shows that the abstraction language admits the same basic algebraic structure as ordinary Boolean polynomials. Moreover, because the constants of the language form the field $\mathbb{B}_{\mathrm{PT}}$, this ring carries a natural scalar multiplication.

\begin{definition}[Scalar Multiplication]
Let $\lambda\in\mathbb{B}_{\mathrm{PT}}$ and $e\in\mathcal{E}$. The scalar multiplication of $e$ by $\lambda$ is defined by $\lambda e := \lambda\odot e.$
\end{definition}

\begin{proposition}[Algebra Structure]
$\mathcal{E}$ is an algebra over the Boolean field of the abstraction language $\mathbb{B}_{\mathrm{PT}}$.
\end{proposition}

\begin{proof}
By the previous proposition, $(\mathcal{E},\uplus,\odot)$ is a commutative ring with unity, and by the preceding result $\mathbb{B}_{\mathrm{PT}}$ is a field. It therefore remains to verify that the scalar multiplication defined above satisfies the compatibility axioms with the ring addition $\uplus$.

The distributive and associative laws of $\odot$ imply that for all $\lambda,\mu\in\mathbb{B}_{\mathrm{PT}}$ and $e,f\in\mathcal{E}$,
\[
(\lambda\uplus\mu)e=\lambda e\uplus\mu e, \qquad
\lambda(e\uplus f)=\lambda e\uplus\lambda f, \qquad
(\lambda\odot\mu)e=\lambda(\mu e),
\]
and the multiplicative identity satisfies $S_I.\mathbf{P}_\emptyset\, e = e$.
Thus scalar multiplication by elements of $\mathbb{B}_{\mathrm{PT}}$ is well defined and satisfies the usual axioms for scalars acting on $\mathcal{E}$.

The distributive laws established in Definition~\ref{def:pt-multiplication} ensure that the multiplication $\odot$ is compatible with the scalar action. Hence, $\mathcal{E}$ is an algebra over $\mathbb{B}_{\mathrm{PT}}$.
\end{proof}

\paragraph{Power Term Polynomial Calculus}
\begin{figure}[t]
\centering
\resizebox{\textwidth}{!}{\begin{tabular}{llll}
\textbf{N°} & \textbf{Conditions on \(S,U,V,T\)} & \textbf{\(S.\mathbf{P}_U\,\odot\,T.\mathbf{P}_V\)} & \(\mathrm{Size}\)\\
\hline
1 & \((S=U=\emptyset)\,\lor\,(T=V=\emptyset)\) & \(S_\emptyset.\mathbf{P}_\emptyset\) & 1\\
2 & \(S=S_I\,\land\,U=\emptyset\) & \(T.\mathbf{P}_V\) & 1\\
3 & \(T=S_I\,\land\,V=\emptyset\) & \(S.\mathbf{P}_U\) & 1\\
\hline
4 & \(U=V=\emptyset\) & \((S\cup T).\mathbf{P}_\emptyset\) & 1\\
5 & \(U=\emptyset\,\land\,S\cap V=\emptyset\) & \((S\cup T).\mathbf{P}_V\) & 1\\
6 & \(V=\emptyset\,\land\,T\cap U=\emptyset\) & \((S\cup T).\mathbf{P}_U\) & 1\\
7 & \(U=\emptyset\,\land\,S\cap V\neq\emptyset\) & \((S\cup T).\mathbf{P}_\emptyset\) & 1\\
8 & \(V=\emptyset\,\land\,T\cap U\neq\emptyset\) & \((S\cup T).\mathbf{P}_\emptyset\) & 1\\
\hline
9 & \(U\cap V=\emptyset\,\land\,S\cap V=\emptyset\,\land T\cap U=\emptyset\) & \((S\cup T).\mathbf{P}_{U\cup V}\uplus(S\cup T).\mathbf{P}_U\uplus(S\cup T).\mathbf{P}_V\) & 3\\
10 & \(U\cap V=\emptyset\,\land\,S\cap V\neq\emptyset\,\land T\cap U=\emptyset\) & \((S\cup T).\mathbf{P}_U\) & 1\\
11 & \(U\cap V=\emptyset\,\land\,S\cap V=\emptyset\,\land T\cap U\neq\emptyset\) & \((S\cup T).\mathbf{P}_V\) & 1\\
12 & \(U\cap V=\emptyset\,\land\,S\cap V\neq\emptyset\,\land T\cap U\neq\emptyset\) & \((S\cup T).\mathbf{P}_\emptyset\) & 1\\
13 & \(U\cap V\neq\emptyset\,\land\,S\cap V\neq\emptyset\,\land T\cap U=\emptyset\) & \((S\cup T).\mathbf{P}_U\) & 1\\
14 & \(U\cap V\neq\emptyset\,\land\,S\cap V=\emptyset\,\land T\cap U\neq\emptyset\) & \((S\cup T).\mathbf{P}_V\) & 1\\
15 & \(U\cap V\neq\emptyset\,\land\,S\cap V\neq\emptyset\,\land T\cap U\neq\emptyset\) & \((S\cup T).\mathbf{P}_\emptyset\) & 1\\
\hline
16 & \(|U\cap V|=1\,\land\,|U\setminus V|>1\,\land\,|V\setminus U|>1\,\land\,|S\cap V|=0\,\land\,|T\cap U|=0\) & \((S\cup T).\mathbf{P}_{U\cup V}\uplus(S\cup T).\mathbf{P}_U\uplus(S\cup T).\mathbf{P}_V\) & 3\\
17 & \(|U\cap V|=1\,\land\,|U\setminus V|>1\,\land\,|V\setminus U|=1\,\land\,|S\cap V|=0\,\land\,|T\cap U|=0\) & \((S\cup T\cup(V\setminus U)).\mathbf{P}_{U\setminus V}\uplus(S\cup T\cup V).\mathbf{P}_{U\setminus V}\uplus(S\cup T\cup (U\cap V)).\mathbf{P}_\emptyset\) & 3\\
18 & \(|U\cap V|=1\,\land\,|U\setminus V|=1\,\land\,|V\setminus U|>1\,\land\,|S\cap V|=0\,\land\,|T\cap U|=0\) & \((S\cup T\cup(U\setminus V)).\mathbf{P}_{V\setminus U}\uplus(S\cup T\cup U).\mathbf{P}_{V\setminus U}\uplus(S\cup T\cup (U\cap V)).\mathbf{P}_\emptyset\) & 3\\
19 & \(|U\cap V|=1\,\land\,|U\setminus V|=1\,\land\,|V\setminus U|=1\,\land\,|S\cap V|=0\,\land\,|T\cap U|=0\) & \((S\cup T\cup(U\setminus V)\cup(V\setminus U)).\mathbf{P}_\emptyset\uplus(S\cup T\cup U\cup V).\mathbf{P}_\emptyset\uplus(S\cup T\cup (U\cap V)).\mathbf{P}_\emptyset\) & 3\\
\hline
20 & \(|U\cap V|>1\,\land\,(|U\setminus V|=0\,\lor\,|V\setminus U|=0)\) & \((S\cup T).\mathbf{P}_{U\cap V}\) & 1\\
21 & \(|U\cap V|>1\,\land\,|U\setminus V|>1\,\land\,|V\setminus U|>1\,\land\,|S\cap V|=0\,\land\,|T\cap U|=0\) & \((S\cup T).\mathbf{P}_{U\cup V}\uplus(S\cup T).\mathbf{P}_U\uplus(S\cup T).\mathbf{P}_V\) & 3\\
22 & \(|U\cap V|>1\,\land\,|U\setminus V|>1\,\land\,|V\setminus U|=1\,\land\,|S\cap V|=0\,\land\,|T\cap U|=0\) & \((S\cup T\cup(V\setminus U)).\mathbf{P}_{U}\uplus(S\cup T\cup (V\setminus U)).\mathbf{P}_{U\cap V}\uplus(S\cup T).\mathbf{P}_{U\cap V}\) & 3\\
23 & \(|U\cap V|>1\,\land\,|U\setminus V|=1\,\land\,|V\setminus U|>1\,\land\,|S\cap V|=0\,\land\,|T\cap U|=0\) & \((S\cup T\cup(U\setminus V)).\mathbf{P}_{V}\uplus(S\cup T\cup (U\setminus V)).\mathbf{P}_{U\cap V}\uplus(S\cup T).\mathbf{P}_{U\cap V}\) & 3\\
24 & \(|U\cap V|>1\,\land\,|U\setminus V|=1\,\land\,|V\setminus U|=1\,\land\,|S\cap V|=0\,\land\,|T\cap U|=0\) & \((S\cup T\cup(U\setminus V)\cup(V\setminus U)).\mathbf{P}_{U\cap V}\uplus(S\cup T).\mathbf{P}_{U\cap V}\uplus(S\cup T\cup(U\setminus V)\cup(V\setminus U)).\mathbf{P}_\emptyset\) & 3\\
\end{tabular}}
\caption{Multiplication rules for any two power terms. Conditions are checked according to the ordering given by the corresponding number. The rules always produce power term polynomials of size no greater than 3. Better visualization by zooming in the figure. A larger version is provided in the Appendix.}
\label{ex:multip}
\end{figure}
We now show that expressions in the abstraction language can be manipulated \emph{syntactically}, without first expanding them into ordinary Boolean polynomials. In particular, the following result provides a fundamental rewriting principle for the calculus: the product of any two power terms can always be transformed into an equivalent power term polynomial using only exclusive disjunction. Thus, multiplication can be eliminated at the syntactic level while preserving semantic equivalence, and the resulting representation remains small.
\begin{theorem}[Multiplication Rewriting Principle]\label{maintheorem}
Let $e,f$ be atomic expressions of the abstraction language,
that is, power terms $S.\mathbf{P}_U,T.\mathbf{P}_V\in\mathcal{T}$
or the constants $S_\emptyset.\mathbf{P}_\emptyset$ and
$S_I.\mathbf{P}_\emptyset$. Then, the product $e\odot f$ is semantically equivalent to a power term polynomial containing at most three atomic terms. Moreover, this polynomial can be obtained by a case analysis on the sets $S,U,T,V$, as summarized in \autoref{ex:multip}.
\end{theorem}
Proof is given in \autoref{app_theorem}. We now demonstrate through an example how the proposed algebra can be used to represent CNF formulae and show how the above rules manipulate terms within the abstraction language by reducing multiplication to sums. This simplifies the original representation and ultimately reveal the ANF representation.
\begin{example}
    Consider the following CNF
    \begin{align}
        \phi:=&(x_1\lor x_2)\land(x_2\lor x_3)\land(\neg x_1\lor\neg x_4\lor x_5)\land(\neg x_1\lor x_4\lor x_6)\land(\neg x_1\lor \neg x_5\lor x_6)\nonumber\\
        &\land(\neg x_1\lor x_4\lor\neg x_6)\land(\neg x_1\lor \neg x_5\lor \neg x_6)     \nonumber  
    \end{align}
    We want first to convert $\phi$ into an equivalent expression within our abstraction language. To do so, we first apply the rules in \autoref{prop:canonical-disjunction} to convert all clauses into their canonical power term polynomial and subsequently combine these polynomials through multiplication. This yields the following expression
    \begin{align}
        e:=&\underbrace{S_\emptyset.\mathbf{P}_{1,2}}_{p_1}\odot\underbrace{S_\emptyset.\mathbf{P}_{2,3}}_{p_2}\odot\underbrace{(S_I.\mathbf{P}_\emptyset\uplus S_{1,4}.\mathbf{P}_\emptyset\uplus S_{1,4,5}.\mathbf{P}_\emptyset)}_{p_3}\odot\underbrace{(S_I.\mathbf{P}_\emptyset\uplus S_{1}.\mathbf{P}_\emptyset\uplus S_{1}.\mathbf{P}_{4,6})}_{p_4} \odot\nonumber\\
        &\underbrace{(S_I.\mathbf{P}_\emptyset\uplus S_{1,5}.\mathbf{P}_\emptyset\uplus S_{1,5,6}.\mathbf{P}_\emptyset)}_{p_5}\odot\underbrace{(S_I.\mathbf{P}_\emptyset\uplus S_{1,6}.\mathbf{P}_\emptyset\uplus S_{1,4,6}.\mathbf{P}_\emptyset)}_{p_6}\odot\underbrace{(S_I.\mathbf{P}_\emptyset\uplus S_{1,5,6}.\mathbf{P}_\emptyset)}_{p_7} \nonumber
    \end{align}
    We can conduct the multiplication between polynomials pairwise. Given the distributive properties of multiplication, each multiplication is applied so that each individual power term in one polynomial is multiplied against each term in the other one. All these individual multiplication are rewritten using the rules in \autoref{maintheorem}.
    Suppose we multiply $p_1$ against $p_3$, $p_5$ against $p_7$ and $p_4$ against $p_6$. The obtained results are:
    \begin{align}
        p_1\odot p_3
        &= \big(S_\emptyset.\mathbf{P}_{1,2}\big)
          \odot
          \big(S_I.\mathbf{P}_\emptyset\uplus S_{1,4}.\mathbf{P}_\emptyset\uplus S_{1,4,5}.\mathbf{P}_\emptyset\big)\nonumber\\
        &= S_{1,4}.\mathbf{P}_\emptyset
          \uplus S_{1,4,5}.\mathbf{P}_\emptyset
          \uplus S_\emptyset.\mathbf{P}_{1,2} \nonumber \\
        p_5\odot p_7
        &=\big(S_I.\mathbf{P}_\emptyset\uplus S_{1,5}.\mathbf{P}_\emptyset\uplus S_{1,5,6}.\mathbf{P}_\emptyset\big)
          \odot
          \big(S_I.\mathbf{P}_\emptyset\uplus S_{1,5,6}.\mathbf{P}_\emptyset\big)\nonumber\\
        &=S_I.\mathbf{P}_\emptyset
          \uplus S_{1,5}.\mathbf{P}_\emptyset \nonumber\\
        p_4\odot p_6
        &=\big(S_I.\mathbf{P}_\emptyset\uplus S_{1}.\mathbf{P}_\emptyset\uplus S_{1}.\mathbf{P}_{4,6}\big)
          \odot
          \big(S_I.\mathbf{P}_\emptyset\uplus S_{1,6}.\mathbf{P}_\emptyset\uplus S_{1,4,6}.\mathbf{P}_\emptyset\big)\nonumber\\
        &=S_I.\mathbf{P}_\emptyset
          \uplus S_{1}.\mathbf{P}_\emptyset
          \uplus S_{1,4}.\mathbf{P}_\emptyset \nonumber
    \end{align}
    The original expression can be simplified into the equivalent form
    \begin{align}
       e=&\underbrace{S_\emptyset.\mathbf{P}_{2,3}}_{p_2}\odot\underbrace{(S_{1,4}.\mathbf{P}_\emptyset
          \uplus S_{1,4,5}.\mathbf{P}_\emptyset
          \uplus S_\emptyset.\mathbf{P}_{1,2})}_{p_1'}\odot\underbrace{(S_I.\mathbf{P}_\emptyset
          \uplus S_{1,5}.\mathbf{P}_\emptyset)}_{p_2'}\odot\underbrace{(S_I.\mathbf{P}_\emptyset
          \uplus S_{1}.\mathbf{P}_\emptyset
          \uplus S_{1,4}.\mathbf{P}_\emptyset)}_{p_3'}\nonumber
    \end{align}
    Suppose that we multiply $p_2$ against $p_2'$ and $p_1'$ against $p_3'$
     \begin{align}
        p_2\odot p_2'
        &=  S_\emptyset.\mathbf{P}_{2,3}\uplus S_{1,5}.\mathbf{P}_{2,3}\nonumber \\
        p_1'\odot p_3'
        &= S_1.\mathbf{P}_\emptyset\uplus S_{1,4,5}.\mathbf{P}_\emptyset\uplus S_\emptyset.\mathbf{P}_{1,2} \underbrace{=}_{\ast} S_2.\mathbf{P}_\emptyset\uplus S_{1,2}.\mathbf{P}_\emptyset\uplus S_{1,4,5}.\mathbf{P}_\emptyset\nonumber
    \end{align}
    where the last equality $\ast$ is obtained by leveraging the decomposition rule in Case 1 of \autoref{shorten_expand} on $S_\emptyset.\mathbf{P}_{1,2}$.
    This yields the following simplified expression
    \begin{align}
        e=&\underbrace{(S_\emptyset.\mathbf{P}_{2,3}\uplus S_{1,5}.\mathbf{P}_{2,3})}_{p_1''}\odot\underbrace{(S_2.\mathbf{P}_\emptyset\uplus S_{1,2}.\mathbf{P}_\emptyset\uplus S_{1,4,5}.\mathbf{P}_\emptyset)}_{p_2''} \nonumber
    \end{align}
    By multiplying $p_1''$ against $p_2''$, we obtain
    \[
    e=S_2.\mathbf{P}_\emptyset\uplus S_{1,2}.\mathbf{P}_\emptyset
    \]
    This can be evaluated to obtain the final ANF representation for $\phi$, that is $\text{ANF}(\phi)=x_2+x_1x_2$.
\end{example}

We provide additional information about the algebra in \autoref{app_info}, including (i) a discussion between power term polynomials and decision diagrams, (ii) exponential blowup and a conversion from power term polynomials to CNF and (iii) an interpretation of the introduced calculus.

\section{Conclusion and Future Work}
\label{sec:conclusions}
We introduced \emph{power term polynomial algebra} as an abstraction language for Boolean polynomials and CNF formulae, motivated by the \emph{tiling mismatch} that arises when moving between resolution-style and algebraic representations. The central contribution of the paper is to show that this mismatch can be addressed at the representation level, without immediately resorting to auxiliary variables and side constraints. Concretely, we defined power terms and power term polynomials, established their semantics, and showed that they form an algebraic structure mirroring Boolean polynomial arithmetic over $\mathbb{F}_2$. We further identified key structural properties of the representation, including compact canonical encodings for disjunctive clauses, local shortening/expansion rules, and a multiplication rewriting principle ensuring that the product of two atomic terms can always be reduced to a power term polynomial of bounded size. Taken together, these results show that the language is expressive enough to encode CNF clauses compactly, structured enough to admit direct symbolic manipulation, and sufficiently algebraic to serve as a bridge between CNF-oriented and ANF-oriented reasoning.

At the same time, the present work is primarily foundational. While the abstraction language avoids some of the immediate overheads of standard CNF$\leftrightarrow$ANF conversions, we have not yet established general complexity bounds for normalization, simplification, or proof search within the language, nor have we provided a complete characterization of when power term polynomial representations are provably more compact than existing CNF- or ANF-based encodings \cite{biere2021handbook}. Likewise, although the multiplication calculus guarantees bounded local rewritings, the global behavior of repeated rewritings remains to be understood: different rewrite orders may lead to substantially different intermediate representations, and minimal representations are in general non-canonical. From a SAT perspective, this means that the current paper should be viewed as introducing a new representation and calculus, rather than yet another competitive solver architecture. In particular, we do not claim empirical superiority over modern CDCL, inprocessing, XOR-aware, or algebraic methods at this stage.

These limitations point to several natural directions for future research. On the theoretical side, an important next step is to study normal forms, confluence properties of the rewrite system, proof complexity connections, and the algorithmic complexity of finding small or optimal power term polynomial representations. It would also be valuable to clarify how this language relates to existing algebraic proof systems and whether it yields new size or simulation results between resolution-style and algebraic derivations.

On the practical side, one promising direction is to use the language as a \emph{preprocessing layer} for SAT solving \cite{biere2021preprocessing}. Since clauses admit compact encodings in this formalism, one can ask whether power term polynomial simplifications can expose structure that is hard to see in raw CNF, and whether this can improve downstream SAT performance after recompilation. Closely related is the question of \emph{conversion}: the abstraction language may provide an intermediate representation for CNF$\rightarrow$ANF and ANF$\rightarrow$CNF transformations that reduces the need for aggressive tiling through auxiliary variables, or at least makes such tiling more principled and structure-aware \cite{choo2019bosphorus,horacek2020conver}.

A further direction is to investigate the language as an \emph{interface} between SAT solvers and algebraic proof systems \cite{choo2019bosphorus,bright2022sat}. Modern workflows often face a sharp boundary between clause-based reasoning and polynomial reasoning; power term polynomials suggest a common layer in which both styles of information can be represented and manipulated. This could enable hybrid procedures in which CNF-derived information and algebraic-derived information are exchanged in a controlled way, rather than through repeated full conversions.

Finally, the abstraction language opens the possibility of designing a new class of \emph{native solvers} that operate directly on power term polynomials. Such solvers would not merely translate back to CNF or ANF, but would branch, simplify, propagate, and rewrite directly in the abstract language. Beyond decision SAT, this perspective also seems relevant to \#SAT and model counting, where compact symbolic representations of structured sets of monomials may offer advantages for counting-based reasoning \cite{chakraborty2021approximate}. Whether these possibilities translate into concrete algorithmic gains remains an open question, but we believe the framework developed here provides a precise foundation on top of which such investigations can now be carried out.

\begin{acknowledgments}
  E.S.\ receives funding from the European Research Council (ERC) under the Horizon Europe research and innovation programme (MSCA-GF grant agreement No. 101149800). A. S.L. received funding from thet National Science Foundation under Grant No. 1918839.
\end{acknowledgments}

\section*{Declaration on Generative AI}
 During the preparation of this work, the author(s) used ChatGPT in order to: Grammar and spelling check. After using these tool(s)/service(s), the author(s) reviewed and edited the content as needed and take(s) full responsibility for the publication’s content. 

\bibliography{main}
\newpage
\appendix

\section{Proof of \autoref{shorten_expand}}\label{app_shorten}
\begin{proposition}[Property 3: Shortening / Expanding Power Term Polynomials (Restated)]
Let $S,T,V\subseteq\mathcal{D}$ with $T\cap V=\emptyset$ and define $U:=T\cup V$ such that $S\cap U=\emptyset$. Then, the following identities hold:
\[
S.\mathbf{P}_U=
\begin{cases}
    (S\cup T).\mathbf{P}_\emptyset\,\uplus\,(S\cup T\cup V).\mathbf{P}_\emptyset\,\uplus\,(S\cup V).\mathbf{P}_\emptyset & \text{if }|T|=1\text{ and }|V|=1\text{ (Case 1)}\\
    S.\mathbf{P}_T\,\uplus\,(S\cup V).\mathbf{P}_T\,\uplus\,(S\cup V).\mathbf{P}_\emptyset & \text{if }|T|>1\text{ and }|V|=1\text{ (Case 2)}\\
   \biguplus_{\widetilde{V}\in\mathbf{P}_V}(S\cup\widetilde{V}).\mathbf{P}_T\,\uplus\, S.\mathbf{P}_T\uplus S.\mathbf{P}_V & \text{if }|T|>1\text{ and }|V|>1\text{ (Case 3)}
\end{cases}
\]
\end{proposition}
\begin{proof}
We prove the three identities by unfolding the semantics of power terms as families of sets. Since \(U=T\cup V\), \(T\cap V=\emptyset\), and \(S\cap U=\emptyset\), it follows that 
\(S\cap T=\emptyset\) and \(S\cap V=\emptyset\). Hence, all power terms occurring below are well formed.

\noindent By Definition~\ref{power_term}, $S.\mathbf{P}_U=\{S\cup W \mid W\in \mathbf{P}_U\}$. Equivalently,
\[
S.\mathbf{P}_U=\{S\cup W \mid \emptyset\neq W\subseteq U\}. \hfill(\star)
\]
We now distinguish the three cases.
\medskip
\begin{itemize}
    \item \textbf{Case 1:} \(|T|=1\) and \(|V|=1\).  $\mathbf{P}_U=\bigl\{T,\,V,\,T\cup V\bigr\}$. Hence,
\[
S.\mathbf{P}_U
=
\{S\cup T,\,S\cup V,\,S\cup T\cup V\}
\]
On the other hand,
\[
(S\cup T).\mathbf{P}_\emptyset=\{S\cup T\},\qquad
(S\cup V).\mathbf{P}_\emptyset=\{S\cup V\},\qquad
(S\cup T\cup V).\mathbf{P}_\emptyset=\{S\cup T\cup V\}.
\]
Since the three families are pairwise disjoint, ordinary union
coincides with exclusive disjunction. Therefore, we can write $S.\mathbf{P}_U
=
(S\cup T).\mathbf{P}_\emptyset
\uplus
(S\cup T\cup V).\mathbf{P}_\emptyset
\uplus
(S\cup V).\mathbf{P}_\emptyset$.
\medskip
\item \textbf{Case 2:} \(|T|>1\) and \(|V|=1\). Let \(V=\{v\}\). Then every non-empty subset \(W\subseteq U=T\cup\{v\}\) belongs to exactly one of the following three disjoint cases:
\smallskip
\begin{enumerate}
    \item \(W\subseteq T\) and \(W\neq\emptyset\);
    \item \(W=\{v\}\);
    \item \(W=A\cup\{v\}\) for some non-empty \(A\subseteq T\).
\end{enumerate}
\smallskip
Therefore, by \((\star)\), the family \(S.\mathbf{P}_U\) decomposes as
\[
S.\mathbf{P}_U
=
\{S\cup A \mid \emptyset\neq A\subseteq T\}
\uplus
\{S\cup\{v\}\}
\uplus
\{S\cup A\cup\{v\}\mid \emptyset\neq A\subseteq T\}.
\]
Now, by noting that
\begin{align}
    &S.\mathbf{P}_T=\{S\cup A \mid \emptyset\neq A\subseteq T\}\nonumber\\
    &(S\cup V).\mathbf{P}_\emptyset=\{S\cup V\}=\{S\cup\{v\}\} \nonumber\\
    &(S\cup V).\mathbf{P}_T=\{S\cup V\cup A \mid \emptyset\neq A\subseteq T\}=\{S\cup A\cup\{v\}\mid \emptyset\neq A\subseteq T\} \nonumber
\end{align}
We obtain
$S.\mathbf{P}_U
=
S.\mathbf{P}_T
\uplus
(S\cup V).\mathbf{P}_T
\uplus
(S\cup V).\mathbf{P}_\emptyset.$
\medskip
\item \textbf{Case 3: \(|T|>1\) and \(|V|>1\).}
Let \(W\subseteq U=T\cup V\) be non-empty. Since \(T\cap V=\emptyset\), \(W\) can be written uniquely as $W=A\cup B,
\qquad A\subseteq T,\quad B\subseteq V$. Because \(W\neq\emptyset\), exactly one of the following occurs:
\smallskip
\begin{enumerate}
    \item \(A\neq\emptyset\) and \(B=\emptyset\);
    \item \(A=\emptyset\) and \(B\neq\emptyset\);
    \item \(A\neq\emptyset\) and \(B\neq\emptyset\).
\end{enumerate}
\smallskip
Thus, by \((\star)\), the family \(S.\mathbf{P}_U\) decomposes as
\[
S.\mathbf{P}_U
=
\{S\cup A \mid \emptyset\neq A\subseteq T\}
\uplus
\{S\cup B \mid \emptyset\neq B\subseteq V\}
\uplus
\{S\cup A\cup B \mid \emptyset\neq A\subseteq T,\ \emptyset\neq B\subseteq V\}.
\]
The first two components are exactly $S.\mathbf{P}_T$ and $S.\mathbf{P}_V$. For the third one, fix \(\widetilde{V}\in\mathbf{P}_V\). Then
\[
(S\cup \widetilde{V}).\mathbf{P}_T
=
\{S\cup \widetilde{V}\cup A \mid \emptyset\neq A\subseteq T\}.
\]
Taking the exclusive disjunction over all non-empty \(\widetilde{V}\subseteq V\) yields
\[
\biguplus_{\widetilde{V}\in\mathbf{P}_V}(S\cup\widetilde{V}).\mathbf{P}_T
=
\{S\cup A\cup B \mid \emptyset\neq A\subseteq T,\ \emptyset\neq B\subseteq V\}.
\]
Therefore, $S.\mathbf{P}_U=\biguplus_{\widetilde{V}\in\mathbf{P}_V}(S\cup\widetilde{V}).\mathbf{P}_T\uplus S.\mathbf{P}_T\uplus S.\mathbf{P}_V$.
\end{itemize}
This proves all three cases.
\end{proof}

\section{Proof of \autoref{maintheorem}}\label{app_theorem}
\begin{sidewaystable}
\centering
\resizebox{\textheight}{!}{
\begin{tabular}{llll}
\textbf{N°} & \textbf{Conditions on \(S,U,V,T\)} & \textbf{\(S.\mathbf{P}_U\,\odot\,T.\mathbf{P}_V\)} & \(\mathrm{Size}\)\\
\hline
1 & \((S=U=\emptyset)\,\lor\,(T=V=\emptyset)\) & \(S_\emptyset.\mathbf{P}_\emptyset\) & 1\\
2 & \(S=S_I\,\land\,U=\emptyset\) & \(T.\mathbf{P}_V\) & 1\\
3 & \(T=S_I\,\land\,V=\emptyset\) & \(S.\mathbf{P}_U\) & 1\\
\hline
4 & \(U=V=\emptyset\) & \((S\cup T).\mathbf{P}_\emptyset\) & 1\\
5 & \(U=\emptyset\,\land\,S\cap V=\emptyset\) & \((S\cup T).\mathbf{P}_V\) & 1\\
6 & \(V=\emptyset\,\land\,T\cap U=\emptyset\) & \((S\cup T).\mathbf{P}_U\) & 1\\
7 & \(U=\emptyset\,\land\,S\cap V\neq\emptyset\) & \((S\cup T).\mathbf{P}_\emptyset\) & 1\\
8 & \(V=\emptyset\,\land\,T\cap U\neq\emptyset\) & \((S\cup T).\mathbf{P}_\emptyset\) & 1\\
\hline
9 & \(U\cap V=\emptyset\,\land\,S\cap V=\emptyset\,\land T\cap U=\emptyset\) & \((S\cup T).\mathbf{P}_{U\cup V}\uplus(S\cup T).\mathbf{P}_U\uplus(S\cup T).\mathbf{P}_V\) & 3\\
10 & \(U\cap V=\emptyset\,\land\,S\cap V\neq\emptyset\,\land T\cap U=\emptyset\) & \((S\cup T).\mathbf{P}_U\) & 1\\
11 & \(U\cap V=\emptyset\,\land\,S\cap V=\emptyset\,\land T\cap U\neq\emptyset\) & \((S\cup T).\mathbf{P}_V\) & 1\\
12 & \(U\cap V=\emptyset\,\land\,S\cap V\neq\emptyset\,\land T\cap U\neq\emptyset\) & \((S\cup T).\mathbf{P}_\emptyset\) & 1\\
13 & \(U\cap V\neq\emptyset\,\land\,S\cap V\neq\emptyset\,\land T\cap U=\emptyset\) & \((S\cup T).\mathbf{P}_U\) & 1\\
14 & \(U\cap V\neq\emptyset\,\land\,S\cap V=\emptyset\,\land T\cap U\neq\emptyset\) & \((S\cup T).\mathbf{P}_V\) & 1\\
15 & \(U\cap V\neq\emptyset\,\land\,S\cap V\neq\emptyset\,\land T\cap U\neq\emptyset\) & \((S\cup T).\mathbf{P}_\emptyset\) & 1\\
\hline
16 & \(|U\cap V|=1\,\land\,|U\setminus V|>1\,\land\,|V\setminus U|>1\,\land\,|S\cap V|=0\,\land\,|T\cap U|=0\) & \((S\cup T).\mathbf{P}_{U\cup V}\uplus(S\cup T).\mathbf{P}_U\uplus(S\cup T).\mathbf{P}_V\) & 3\\
17 & \(|U\cap V|=1\,\land\,|U\setminus V|>1\,\land\,|V\setminus U|=1\,\land\,|S\cap V|=0\,\land\,|T\cap U|=0\) & \((S\cup T\cup(V\setminus U)).\mathbf{P}_{U\setminus V}\uplus(S\cup T\cup V).\mathbf{P}_{U\setminus V}\uplus(S\cup T\cup (U\cap V)).\mathbf{P}_\emptyset\) & 3\\
18 & \(|U\cap V|=1\,\land\,|U\setminus V|=1\,\land\,|V\setminus U|>1\,\land\,|S\cap V|=0\,\land\,|T\cap U|=0\) & \((S\cup T\cup(U\setminus V)).\mathbf{P}_{V\setminus U}\uplus(S\cup T\cup U).\mathbf{P}_{V\setminus U}\uplus(S\cup T\cup (U\cap V)).\mathbf{P}_\emptyset\) & 3\\
19 & \(|U\cap V|=1\,\land\,|U\setminus V|=1\,\land\,|V\setminus U|=1\,\land\,|S\cap V|=0\,\land\,|T\cap U|=0\) & \((S\cup T\cup(U\setminus V)\cup(V\setminus U)).\mathbf{P}_\emptyset\uplus(S\cup T\cup U\cup V).\mathbf{P}_\emptyset\uplus(S\cup T\cup (U\cap V)).\mathbf{P}_\emptyset\) & 3\\
\hline
20 & \(|U\cap V|>1\,\land\,(|U\setminus V|=0\,\lor\,|V\setminus U|=0)\) & \((S\cup T).\mathbf{P}_{U\cap V}\) & 1\\
21 & \(|U\cap V|>1\,\land\,|U\setminus V|>1\,\land\,|V\setminus U|>1\,\land\,|S\cap V|=0\,\land\,|T\cap U|=0\) & \((S\cup T).\mathbf{P}_{U\cup V}\uplus(S\cup T).\mathbf{P}_U\uplus(S\cup T).\mathbf{P}_V\) & 3\\
22 & \(|U\cap V|>1\,\land\,|U\setminus V|>1\,\land\,|V\setminus U|=1\,\land\,|S\cap V|=0\,\land\,|T\cap U|=0\) & \((S\cup T\cup(V\setminus U)).\mathbf{P}_{U}\uplus(S\cup T\cup (V\setminus U)).\mathbf{P}_{U\cap V}\uplus(S\cup T).\mathbf{P}_{U\cap V}\) & 3\\
23 & \(|U\cap V|>1\,\land\,|U\setminus V|=1\,\land\,|V\setminus U|>1\,\land\,|S\cap V|=0\,\land\,|T\cap U|=0\) & \((S\cup T\cup(U\setminus V)).\mathbf{P}_{V}\uplus(S\cup T\cup (U\setminus V)).\mathbf{P}_{U\cap V}\uplus(S\cup T).\mathbf{P}_{U\cap V}\) & 3\\
24 & \(|U\cap V|>1\,\land\,|U\setminus V|=1\,\land\,|V\setminus U|=1\,\land\,|S\cap V|=0\,\land\,|T\cap U|=0\) & \((S\cup T\cup(U\setminus V)\cup(V\setminus U)).\mathbf{P}_{U\cap V}\uplus(S\cup T).\mathbf{P}_{U\cap V}\uplus(S\cup T\cup(U\setminus V)\cup(V\setminus U)).\mathbf{P}_\emptyset\) & 3\\
\end{tabular}}
\caption{(Same as \autoref{ex:multip}) Multiplication rules for any two power terms. Conditions are checked according to the ordering given by the corresponding number. The rules always produce power term polynomials of size no greater than 3.}
\label{ex:multip2}
\end{sidewaystable}
\begin{theorem}[Multiplication Rewriting Principle (Restated)]
Let $e,f$ be atomic expressions of the abstraction language,
that is, power terms $S.\mathbf{P}_U,T.\mathbf{P}_V\in\mathcal{T}$
or the constants $S_\emptyset.\mathbf{P}_\emptyset$ and
$S_I.\mathbf{P}_\emptyset$. Then, the product $e\odot f$ is semantically equivalent to a power term polynomial containing at most three atomic terms. Moreover, this polynomial can be obtained by a case analysis on the sets $S,U,T,V$, as summarized in \autoref{ex:multip2}.
\end{theorem}
\begin{proof}
For a finite set $W\subseteq\mathcal D$, write
\[
m_W:=\prod_{x\in W}x
\qquad\text{and}\qquad
F_W:=\sum_{\emptyset\neq A\subseteq W} m_A.
\]
Thus, for every well-formed power term $S.\mathbf{P}_U$,
\[
\meval{S.\mathbf{P}_U}=m_S\,F_U,
\]
while the two constants satisfy
\[
\meval{S_\emptyset.\mathbf{P}_\emptyset}=0,
\qquad
\meval{S_I.\mathbf{P}_\emptyset}=1.
\]

We prove the theorem by a case analysis on the atomic factors
$e$ and $f$.

\medskip
\noindent\textbf{Step 1: Constant cases (Rules 1--3).}
If one factor is $S_\emptyset.\mathbf{P}_\emptyset$, then the product is
$0$, hence semantically equal to $S_\emptyset.\mathbf{P}_\emptyset$.
If one factor is $S_I.\mathbf{P}_\emptyset$, then the product is the other
factor. This gives Rules~1--3.

\medskip
\noindent\textbf{Step 2: One or both factors are degree-$0$ power terms (Rules 4--8).}
Assume first that $U=V=\emptyset$. Then
\[
\meval{S.\mathbf{P}_\emptyset\odot T.\mathbf{P}_\emptyset}
= m_Sm_T
= m_{S\cup T},
\]
so
\[
S.\mathbf{P}_\emptyset\odot T.\mathbf{P}_\emptyset
\equiv (S\cup T).\mathbf{P}_\emptyset,
\]
which is Rule~4.

Now assume $U=\emptyset$ and $V\neq\emptyset$. Then
\[
\meval{S.\mathbf{P}_\emptyset\odot T.\mathbf{P}_V}
= m_Sm_TF_V
= m_{S\cup T}F_V,
\]
provided $S\cap V=\emptyset$, which yields
\[
S.\mathbf{P}_\emptyset\odot T.\mathbf{P}_V
\equiv (S\cup T).\mathbf{P}_V.
\]
This is Rule~5.

If instead $U=\emptyset$ and $S\cap V\neq\emptyset$, let $x\in S\cap V$.
Since $x$ divides $m_S$, we have in the Boolean polynomial ring
\[
m_S\prod_{v\in V}(1+v)=0,
\]
because the factor $x(1+x)=x+x^2=0$. Using
\[
F_V=\prod_{v\in V}(1+v)+1,
\]
it follows that
\[
m_SF_V = m_S.
\]
Hence
\[
\meval{S.\mathbf{P}_\emptyset\odot T.\mathbf{P}_V}
= m_T(m_SF_V)
= m_{S\cup T},
\]
so
\[
S.\mathbf{P}_\emptyset\odot T.\mathbf{P}_V
\equiv (S\cup T).\mathbf{P}_\emptyset,
\]
which is Rule~7. Rule~6 and Rule~8 are symmetric.

\medskip
\noindent\textbf{Step 3: Both factors have non-empty power parts (Rules 10--15).}
From now on assume that $U\neq\emptyset$ and $V\neq\emptyset$.

\smallskip
\noindent\emph{Reduction when the base of one factor meets the power part of the other.}
If $S\cap V\neq\emptyset$, the same argument as above gives
$m_SF_V=m_S$, hence
\[
\meval{S.\mathbf{P}_U\odot T.\mathbf{P}_V}
= m_T(m_SF_V)F_U
= m_{S\cup T}F_U
= \meval{(S\cup T).\mathbf{P}_U}.
\]
Therefore
\[
S.\mathbf{P}_U\odot T.\mathbf{P}_V
\equiv (S\cup T).\mathbf{P}_U
\qquad\text{if }S\cap V\neq\emptyset\text{ and }T\cap U=\emptyset,
\]
which yields Rules~10 and~13 according as $U\cap V=\emptyset$ or
$U\cap V\neq\emptyset$.

Similarly, if $S\cap V=\emptyset$ and $T\cap U\neq\emptyset$, then
\[
\meval{S.\mathbf{P}_U\odot T.\mathbf{P}_V}
= \meval{(S\cup T).\mathbf{P}_V},
\]
hence
\[
S.\mathbf{P}_U\odot T.\mathbf{P}_V
\equiv (S\cup T).\mathbf{P}_V,
\]
which gives Rules~11 and~14.

Finally, if both $S\cap V\neq\emptyset$ and $T\cap U\neq\emptyset$, then
both collapses occur and
\[
\meval{S.\mathbf{P}_U\odot T.\mathbf{P}_V}=m_{S\cup T},
\]
hence
\[
S.\mathbf{P}_U\odot T.\mathbf{P}_V
\equiv (S\cup T).\mathbf{P}_\emptyset.
\]
This gives Rules~12 and~15.

\smallskip
\noindent\emph{Core case: }$S\cap V=\emptyset$ and $T\cap U=\emptyset$.
In this case no collapse occurs, so
\[
\meval{S.\mathbf{P}_U\odot T.\mathbf{P}_V}
= m_{S\cup T}\,F_UF_V.
\]
We now use the identity
\begin{equation}\label{eq:FU-FV}
F_UF_V = F_{U\cup V}+F_U+F_V,
\end{equation}
where $+$ is addition in the Boolean polynomial ring over
$\mathbb F_2$.
Indeed,
\[
1+F_W=\prod_{w\in W}(1+w)
\]
for every finite $W\subseteq\mathcal D$, and therefore
\[
(1+F_U)(1+F_V)
=\prod_{u\in U}(1+u)\prod_{v\in V}(1+v)
=\prod_{z\in U\cup V}(1+z)
=1+F_{U\cup V},
\]
because $(1+x)^2=1+x$ in the Boolean polynomial ring. Expanding the
left-hand side over $\mathbb F_2$ gives \eqref{eq:FU-FV}.

Multiplying \eqref{eq:FU-FV} by $m_{S\cup T}$ yields
\begin{equation}\label{eq:raw-product}
\meval{S.\mathbf{P}_U\odot T.\mathbf{P}_V}
=
\meval{(S\cup T).\mathbf{P}_{U\cup V}}
+
\meval{(S\cup T).\mathbf{P}_U}
+
\meval{(S\cup T).\mathbf{P}_V}.
\end{equation}
Thus the product is always semantically equivalent to a power term
polynomial with at most three atomic terms. It remains to rewrite
\eqref{eq:raw-product} into well-formed syntax according to the possible
shapes of $U$ and $V$.

\medskip
\noindent\textbf{Step 4: Syntactic normalization of the core case (Rules 9, 16--24).}

Starting from \eqref{eq:raw-product}, set
\[
A:=U\setminus V,\qquad B:=V\setminus U,\qquad C:=U\cap V.
\]
Then \(U=A\cup C\), \(V=B\cup C\), and \(A,B,C\) are pairwise disjoint.
Moreover, in the present core case we have \(S\cap V=\emptyset\) and
\(T\cap U=\emptyset\), hence also
\[
S\cap(A\cup B\cup C)=\emptyset,\qquad
T\cap(A\cup B\cup C)=\emptyset.
\]
Thus every power term appearing below has disjoint base and power part.

We distinguish the possible cardinalities of \(A,B,C\).

\smallskip
\noindent\emph{(i) Disjoint case: \(C=\emptyset\).}
Then \(U=A\), \(V=B\), and \eqref{eq:raw-product} becomes
\[
(S\cup T).\mathbf{P}_{A\cup B}
\uplus
(S\cup T).\mathbf{P}_A
\uplus
(S\cup T).\mathbf{P}_B.
\]
If \(|A|>1\) and \(|B|>1\), all three terms are well formed, giving
Rule~9. The cases \(|A|=1\) or \(|B|=1\) do not occur here, because
\(U\) and \(V\) are themselves power parts of well-formed power terms and
therefore cannot have size \(1\).

\smallskip
\noindent\emph{(ii) Inclusion case: \(C\neq\emptyset\) and \(A=\emptyset\) or \(B=\emptyset\).}
If \(A=\emptyset\), then \(U=C\subseteq V\), so from \eqref{eq:raw-product}
\[
F_{U\cup V}+F_U+F_V = F_V+F_U+F_V = F_U.
\]
Hence
\[
\meval{S.\mathbf{P}_U\odot T.\mathbf{P}_V}
=\meval{(S\cup T).\mathbf{P}_U},
\]
which is Rule~20. The case \(B=\emptyset\) is symmetric.

\smallskip
\noindent\emph{(iii) Proper overlap case: \(C\neq\emptyset\), \(A\neq\emptyset\), \(B\neq\emptyset\).}
Now \eqref{eq:raw-product} reads
\begin{equation}\label{eq:step4-overlap}
(S\cup T).\mathbf{P}_{A\cup B\cup C}
\uplus
(S\cup T).\mathbf{P}_{A\cup C}
\uplus
(S\cup T).\mathbf{P}_{B\cup C}.
\end{equation}
This always gives the correct semantics; the only issue is that some
power parts in \eqref{eq:step4-overlap} may have size \(1\), which is not
well formed. We eliminate such singleton power parts using
Proposition~\ref{shorten_expand}.

\smallskip
\noindent\textbf{Case \(|C|=1\).}
Write \(C=\{c\}\).

If \(|A|>1\) and \(|B|>1\), then all three power parts
\(A\cup B\cup C\), \(A\cup C\), \(B\cup C\) have size \(>1\), so
\eqref{eq:step4-overlap} is already well formed. This is Rule~16.

If \(|A|>1\) and \(|B|=1\), then \(B\cup C\) has size \(2\), while
\(A\cup C\) and \(A\cup B\cup C\) have size \(>1\). Applying
Proposition~\ref{shorten_expand} (Case~2 there) to the first term of
\eqref{eq:step4-overlap}, with base \(S\cup T\), \(T'=A\), \(V'=B\cup C\),
gives
\[
(S\cup T).\mathbf{P}_{A\cup B\cup C}
=
(S\cup T).\mathbf{P}_A
\uplus
(S\cup T\cup(B\cup C)).\mathbf{P}_A
\uplus
(S\cup T\cup(B\cup C)).\mathbf{P}_\emptyset.
\]
Since also
\[
(S\cup T).\mathbf{P}_{A\cup C}
=
(S\cup T).\mathbf{P}_A
\uplus
(S\cup T\cup C).\mathbf{P}_A
\uplus
(S\cup T\cup C).\mathbf{P}_\emptyset
\]
by the same proposition, the two copies of \((S\cup T).\mathbf{P}_A\)
cancel under \(\uplus\). Because \(B\cup C=V\) and \(C=U\cap V\), this yields
\[
(S\cup T\cup(V\setminus U)).\mathbf{P}_{U\setminus V}
\uplus
(S\cup T\cup V).\mathbf{P}_{U\setminus V}
\uplus
(S\cup T\cup(U\cap V)).\mathbf{P}_\emptyset,
\]
which is Rule~17.

The case \(|A|=1\), \(|B|>1\) is symmetric and gives Rule~18.

If \(|A|=|B|=1\), then each of the three power parts in
\eqref{eq:step4-overlap} has size \(2\) or \(3\), and applying
Proposition~\ref{shorten_expand} (Case~1 there) to
\((S\cup T).\mathbf{P}_{A\cup B\cup C}\) and cancelling the repeated
degree-\(0\) term gives
\[
(S\cup T\cup A\cup B).\mathbf{P}_\emptyset
\uplus
(S\cup T\cup A\cup B\cup C).\mathbf{P}_\emptyset
\uplus
(S\cup T\cup C).\mathbf{P}_\emptyset,
\]
that is,
\[
(S\cup T\cup(U\setminus V)\cup(V\setminus U)).\mathbf{P}_\emptyset
\uplus
(S\cup T\cup U\cup V).\mathbf{P}_\emptyset
\uplus
(S\cup T\cup(U\cap V)).\mathbf{P}_\emptyset,
\]
which is Rule~19.

\smallskip
\noindent\textbf{Case \(|C|>1\).}

If \(|A|>1\) and \(|B|>1\), then all three terms in
\eqref{eq:step4-overlap} are already well formed, giving Rule~21.

If \(|A|>1\) and \(|B|=1\), apply Proposition~\ref{shorten_expand}
(Case~2 there) to the first term of \eqref{eq:step4-overlap}, now with
base \(S\cup T\), \(T'=U=A\cup C\), \(V'=B\). This gives
\[
(S\cup T).\mathbf{P}_{A\cup B\cup C}
=
(S\cup T).\mathbf{P}_U
\uplus
(S\cup T\cup B).\mathbf{P}_U
\uplus
(S\cup T\cup B).\mathbf{P}_\emptyset.
\]
Since \((S\cup T).\mathbf{P}_U\) cancels with the second term of
\eqref{eq:step4-overlap}, we obtain
\[
(S\cup T\cup(V\setminus U)).\mathbf{P}_U
\uplus
(S\cup T\cup(V\setminus U)).\mathbf{P}_{U\cap V}
\uplus
(S\cup T).\mathbf{P}_{U\cap V},
\]
which is Rule~22.

The case \(|A|=1\), \(|B|>1\) is symmetric and gives Rule~23.

Finally, if \(|A|=|B|=1\), apply Proposition~\ref{shorten_expand}
(Case~2 there) first to the term \((S\cup T).\mathbf{P}_{A\cup B\cup C}\)
with \(T'=C\), \(V'=A\cup B\), and then use Case~1 of the same proposition
on the degree-\(2\) part \(A\cup B\). After cancellation with the two
remaining terms in \eqref{eq:step4-overlap}, one obtains
\[
(S\cup T\cup A\cup B).\mathbf{P}_C
\uplus
(S\cup T).\mathbf{P}_C
\uplus
(S\cup T\cup A\cup B).\mathbf{P}_\emptyset,
\]
that is,
\[
(S\cup T\cup(U\setminus V)\cup(V\setminus U)).\mathbf{P}_{U\cap V}
\uplus
(S\cup T).\mathbf{P}_{U\cap V}
\uplus
(S\cup T\cup(U\setminus V)\cup(V\setminus U)).\mathbf{P}_\emptyset,
\]
which is Rule~24.

\smallskip
This exhausts all possibilities in the core case. Therefore
\eqref{eq:raw-product} can always be rewritten into one of Rules~9,
16--24, and in each case the result is a well-formed power term
polynomial with at most three atomic terms.

Hence, for any two atomic expressions $e$ and $f$, the product
$e\odot f$ is semantically equivalent to a power term polynomial of size
at most $3$, obtained by the case analysis summarized in
\autoref{ex:multip2}.
\end{proof}

\section{Additional Information about the Algebra}\label{app_info}
We discuss three additional properties of the proposed algebra using intuitive examples:
\begin{enumerate}
    \item Comparison of representations, including Power Term Polynomials (PTP) and Decision Diagrams (DD).
    \item Notes about conversions, exponential blowup and PTP$\to$CNF conversion.
    \item Interpretation of PTP multiplication rules.
\end{enumerate}

\subsection*{Part 1: PTP vs. DD}

PTPs and DDs are fundamentally different representations with distinct purposes. DDs belong to the family of negation normal form (NNF). They are obtained by introducing additional properties over NNF, such as decomposability, determinism, and smoothness. These properties support tractable inference. In contrast, PTPs are a generalization of algebraic normal form (ANF). PTPs provide a more succinct representation of ANF by grouping monomials. They are designed to bridge CNF and ANF representations while avoiding exponential blowup. Moreover, PTPs admit a calculus that can be used for simplification, as shown later.

As a clarifying example about their difference, consider expressing a clause, that is, a disjunction of $n$ literals, using PTPs and DDs. According to Proposition~15, PTPs encode the clause with at most $3$ power terms. In contrast, DDs require $n$ decision nodes, highlighting their dependence on the number of variables/literals in the clause.

\subsection*{Part 2: Conversions and Exponential Blowup}

In the paper, we have already described the advantage of conversion CNF$\to$PTP over CNF$\to$ANF. The former produces one PTP for each clause and each PTP is compactly described by $3$ power terms. In contrast, the latter often produces exponentially large polynomials for each clause and introduces additional polynomials to constrain the behaviour of auxiliary variables; see Example~2.

We have also explained the conversion ANF$\to$CNF; see Example~3. This is carried out in four stages:
\begin{enumerate}
    \item The ANF polynomial is first linearised, meaning that an auxiliary variable is introduced for each monomial with degree larger than $1$.
    \item The linearised polynomial is further split into equally sized chunks/tiles by introducing extra auxiliary variables. This is done to keep the final conversion in the third stage tractable.
    \item Each tile is converted to CNF through the ANF$\to$Boolean$\to$CNF conversion.
    \item The resulting CNF from the tiles is conjoined with the CNF obtained from the constraints on the auxiliary variables introduced in Stages~1 and~2.
\end{enumerate}

Importantly, PTP$\to$CNF can be obtained by modifying Stage~1 of ANF$\to$CNF. The key idea is to linearise PTPs at the level of power terms instead of monomials. Specifically, we introduce an auxiliary variable for each power term grouping more than one monomial. This exponentially reduces the number of terms in the linearised polynomial, as well as auxiliary variables and constraints, yielding exponential savings in the number of clauses and variables in the final CNF.

We now provide an example.

\paragraph{Example.}
Suppose we have the following ANF and equivalent PTP, where we use parentheses for better readability:
\[
\mathrm{ANF}
=
(x_1+x_2+x_3+x_4)
+
(x_1x_2+x_1x_3+x_1x_4)
+
(x_1x_2x_3+x_1x_2x_4+x_1x_3x_4)
+
x_1x_2x_3x_4,
\]
and
\[
\mathrm{PTP}
=
S_1.\mathbf{P}_\emptyset
\uplus
S_2.\mathbf{P}_\emptyset
\uplus
S_3.\mathbf{P}_\emptyset
\uplus
S_4.\mathbf{P}_\emptyset
\uplus
S_1.\mathbf{P}_{234}.
\]

Stage~1 of ANF$\to$CNF produces a linearised polynomial with $11$ degree-$1$ monomials and introduces $7$ auxiliary variables with corresponding constraints, for example $x_1' \leftrightarrow x_1x_2$.

In contrast, Stage~1 for PTP produces a linearised polynomial with $5$ terms and a single auxiliary variable. That is,
\[
p = x_1+x_2+x_3+x_4+x_1',
\]
where $p$ is the linearised polynomial and
\[
x_1' \Leftrightarrow x_1 \land (x_2 \lor x_3 \lor x_4).
\]
Note how $S_1.\mathbf{P}_{234}$ can be expressed in terms of conjunction and disjunction operations thanks to Proposition~15. Clearly, $p$ can now be processed through standard ANF$\to$CNF conversion.

Moreover, the constraint
\[
x_1' \Leftrightarrow x_1 \land (x_2 \lor x_3 \lor x_4)
\]
can be easily converted into a CNF using a number of clauses that is at most linear in the number of variables involved in the clause. Indeed,
\[
x_1' \Leftarrow x_1 \land (x_2 \lor x_3 \lor x_4)
\]
is equivalently rewritten using three clauses:
\[
(x_1' \lor \neg x_1 \lor \neg x_2)
\land
(x_1' \lor \neg x_1 \lor \neg x_3)
\land
(x_1' \lor \neg x_1 \lor \neg x_4),
\]
and
\[
x_1' \Rightarrow x_1 \land (x_2 \lor x_3 \lor x_4)
\]
results in two clauses:
\[
(\neg x_1' \lor x_1)
\land
(\neg x_1' \lor x_2 \lor x_3 \lor x_4).
\]

Therefore, CNF and PTP representations can be translated without incurring exponential blowup, unlike ANFs.

\subsection*{Part 3: Interpretation of PTP Multiplication Rules}

We highlight how most multiplication rules in Theorem~25 eliminate groups of monomials arising from combinations of the same variables. These rules can be interpreted as symbolic analogues of subsumption and variable elimination, similarly to known preprocessing strategies for CNFs, such as bounded variable elimination, but operating directly at the level of grouped monomials.

The rules in Figure~3 can be organised into five mutually exclusive cases:
\begin{itemize}
    \item \textbf{Case 1:}
    $
    (U=\emptyset) \lor (V=\emptyset),
    $
    including Rules~1 to~8, single monomial vs. grouped monomials.

    \item \textbf{Case 2:}
    $
    (U \cap V = \emptyset) \land \neg \text{Case 1},
    $
    including Rules~9 to~12, namely multiplications between power terms with some overlaps between variables.

    \item \textbf{Case 3:}
    $(U \cap V \neq \emptyset)
    \land
    \bigl((S \cap V \neq \emptyset) \lor (T \cap U \neq \emptyset)\bigr),$
    including Rules~13 to~15, namely multiplications between power terms with ``exponential'' overlap between variables.

    \item \textbf{Case 4:}
    $
    (U \cap V \neq \emptyset)
    \land
    \bigl((S \cap V = \emptyset) \lor (T \cap U = \emptyset)\bigr)
    \land
    \bigl((U \setminus V = \emptyset) \lor (V \setminus U = \emptyset)\bigr),
    $
    including Rule~20.

    \item \textbf{Case 5:} the remaining case.
\end{itemize}

We now show some examples illustrating how the rules exploit subsumption to eliminate variables, valid for all Cases~1--4 except Rule~9.

\paragraph{Examples.}

\textbf{Rule 7, Case 1.}
Consider
\[
S_1.\mathbf{P}_\emptyset \odot S_\emptyset.\mathbf{P}_{123} = S_1.\mathbf{P}_\emptyset.
\]
By Proposition~15, $S_1.\mathbf{P}_\emptyset \odot S_\emptyset.\mathbf{P}_{123}$ is equivalent to
\[
x_1 \land (x_1 \lor x_2 \lor x_3).
\]
Now, $x_1 \land (x_1 \lor x_2 \lor x_3)$
can be replaced by $x_1$ through subsumption.

\textbf{Rule 10, Case 2.}
Consider
\[
S_1.\mathbf{P}_{23} \odot T_4.\mathbf{P}_{1567} = S_{14}.\mathbf{P}_{23}.
\]
Equivalently, we have
\[
x_1 \land (x_2 \lor x_3)
\land
x_4 \land (x_1 \lor x_5 \lor x_6 \lor x_7).
\]
This can be replaced by
\[
x_1 \land x_4 \land (x_2 \lor x_3).
\]

\textbf{Rule 15, Case 3.}
Consider
\[
S_1.\mathbf{P}_{234} \odot T_2.\mathbf{P}_{135} = S_{12}.\mathbf{P}_\emptyset.
\]
That is,
\[
x_1 \land (x_2 \lor x_3 \lor x_4)
\land
x_2 \land (x_1 \lor x_3 \lor x_5)
\]
is replaced by
\[
x_1 \land x_2.
\]

\textbf{Rule 20, Case 4.}
\[
S_\emptyset.\mathbf{P}_{1234} \odot S_\emptyset.\mathbf{P}_{12} = S_\emptyset.\mathbf{P}_{12}.
\]

Beyond these local effects, the calculus also captures non-trivial interactions between structured clauses that are not reducible by standard CNF preprocessing.

\paragraph{Structured interaction example.}
Consider
\[
(x_1 \lor x_2 \lor x_3)
\land
(x_1 \lor x_2 \lor x_4).
\]
In PTP, this corresponds to
\[
S_\emptyset.\mathbf{P}_{123} \odot S_\emptyset.\mathbf{P}_{124},
\]
which falls under Rule~24 and rewrites to
\[
S_{34}.\mathbf{P}_{12} \uplus S_\emptyset.\mathbf{P}_{12} \uplus S_{34}.\mathbf{P}_\emptyset.
\]
This decomposition captures the interaction between the two clauses through their shared variables $x_1,x_2$, while isolating their differences $x_3$ and $x_4$. Importantly, this transformation is not obtained through local CNF simplifications such as subsumption or unit propagation, and avoids enumerating the full monomial expansion that would arise in direct ANF.

\end{document}